\documentclass[11pt,a4paper]{amsart}
\usepackage{fullpage}
\usepackage{graphicx}
\usepackage{subcaption}
\usepackage{color}
\usepackage{amsmath}
\usepackage{amssymb}
\usepackage{amscd}
\usepackage{bbm}
\usepackage{float}
\newcommand{\R}{\mathbb{R}}
\newcommand{\inr}[1]{\bigl< #1 \bigr>}

\newcommand{\N}{\mathbb{N}}

\newcommand{\E}{\mathbb{E}}

\newcommand{\eps}{\varepsilon}
\newcommand{\conv}{\mathop{\rm conv}}

\newtheorem{Theorem}{Theorem}[section]
\newtheorem{Lemma}[Theorem]{Lemma}

\newtheorem{Definition}[Theorem]{Definition}

\newtheorem{Corollary}[Theorem]{Corollary}
\newtheorem{Remark}[Theorem]{Remark}

\numberwithin{equation}{section}

\def \proof {\noindent {\bf Proof.}\ \ }

\def \endproof
{{\mbox{}\nolinebreak\hfill\rule{2mm}{2mm}\par\medbreak}}
\def\IND{\mathbbm{1}}

\allowdisplaybreaks

\newcommand{\la}{\lambda}

\newcommand{\al}{\alpha}

\newcommand{\bP}{\mathbb{P}}
\newcommand{\ti}{\times}
\newcommand{\si}{\sigma}
\newcommand{\cU}{\mathcal{U}}
\newcommand{\thres}{\tau_{\operatorname{thres}}}
\newcommand{\noise}{\nu_{\operatorname{noise}}}
\newcommand{\Si}{\Sigma}
\newcommand{\sign}{\operatorname{sign}}
\newcommand{\co}{\operatorname{conv}}

\newcommand{\red}{}
\newcommand{\blue}{}
\newcommand{\green}{}

\title{Non-Gaussian Hyperplane Tessellations and Robust One-Bit Compressed Sensing}
\author{Sjoerd Dirksen}
\thanks{Lehrstuhl C f{\"u}r
Mathematik (Analysis), RWTH Aachen University, dirksen@mathc.rwth-aachen.de}

\author{Shahar Mendelson}

\thanks{Mathematical Sciences Institute, The Australian National University and Department of Mathematics,\\ \indent Technion, I.I.T, shahar.mendelson@gmail.com}

\begin{document}

\maketitle

\begin{abstract}
\noindent We show that a tessellation generated by a small number of random affine hyperplanes can be used to approximate Euclidean distances between any two points in an arbitrary bounded set $T$, where the random hyperplanes are generated by subgaussian or heavy-tailed normal vectors and uniformly distributed shifts. We derive quantitative bounds on the number of hyperplanes needed for constructing such tessellations in terms of natural metric complexity measures of $T$ and the desired approximation error. Our work extends significantly prior results in this direction, which were restricted to Gaussian hyperplane tessellations of subsets of the Euclidean unit sphere.

As an application, we obtain new reconstruction results in memoryless one-bit compressed sensing with non-Gaussian measurement matrices. We show that by quantizing at uniformly distributed thresholds, it is possible to accurately reconstruct low-complexity signals from a small number of one-bit quantized measurements, even if the measurement vectors are drawn from a heavy-tailed distribution. Our reconstruction results are uniform in nature and robust in the presence of pre-quantization noise on the analog measurements as well as adversarial bit corruptions in the quantization process. Moreover we show that if the measurement matrix is subgaussian then accurate recovery can be achieved via a convex program.
\end{abstract}

\section{Introduction}

In this article we study a fundamental geometric question: can distances between points in a given set $T \subset \R^n$ be accurately encoded using a small number of random hyperplanes? To formulate this question more precisely, let $H_{X_i,\tau_i} = \{x\in \R^n \ : \ \langle X_i,x\rangle +\tau_i=0\}$, $i=1,\ldots,m$, be a set of affine hyperplanes with normal vectors $X_i$ and shift parameters $\tau_i$. These hyperplanes tessellate the set $T$ into (at most) $2^m$ cells and, for any $x\in T$, the bit string $(\sign(\langle X_i,x\rangle +\tau_i))_{i=1}^m \in \{-1,1\}^m$ encodes the cell in which $x$ is located (see Figures~\ref{fig:firstCut} and \ref{fig:secondCut}).  Moreover, for any two points $x,y\in T$, the \green{normalized} Hamming distance between their bit strings
\begin{equation} \label{eq:hamming}
\frac{1}{m}|\{i: \sign(\langle X_i,x\rangle+\tau_i) \not=\sign(\langle X_i,y\rangle+\tau_i)\}|
\end{equation}
counts the \green{fraction} of hyperplanes separating $x$ and $y$. In what follows we are interested in quantifying the number of random hyperplanes that suffice to ensure that \eqref{eq:hamming}
approximates the distance between any two points in $T$ that are not `too close'.
\par
A beautiful result due to Plan and Vershynin \cite{PlV14} essentially solves this question for subsets of the Euclidean unit sphere with respect to the geodesic distance, using homogeneous Gaussian hyperplanes (i.e., $\tau_i=0$ for all $i$). They showed that if $T\subset S^{n-1}$ and the normal vectors $X_1,\ldots X_m$ are independent standard Gaussian vectors, then with probability at least $1-2e^{-cm\rho^2}$, for all $x,y\in T$,
\begin{equation}
\label{eqn:PlVEmbed}
d_{S^{n-1}}(x,y) - \rho \leq \frac{1}{m}|\{i: \sign(\langle X_i,x\rangle) \not=\sign(\langle X_i,y\rangle)\}| \leq d_{S^{n-1}}(x,y) + \rho,
\end{equation}
provided that $m\gtrsim \rho^{-6} \ell_*^2(T)$; here
$$
\ell_*(T) : = \E\sup_{x\in T}|\langle G,x\rangle|
$$
where $G$ is the standard Gaussian random vector in $\R^n$. Thus, $\ell_*(T)$ is the Gaussian mean width of $T$---a natural geometric parameter that is of central importance in geometry (e.g. Dvoretzky type theorems, see for instance \cite{AGM15}) and in statistics, where it is used to capture the difficulty of prediction problems in numerous manuscripts.

It follows from \eqref{eqn:PlVEmbed} that if $x$ and $y$ are `far enough apart', then the fraction of homogeneous Gaussian hyperplanes that separate them concentrates sharply around their geodesic distance.

As far as random homogeneous Gaussian tessellations of $T \subset S^{n-1}$ are concerned, it was conjectured in \cite{PlV14} that $m\simeq \rho^{-2}\ell_*^2(T)$ is necessary and sufficient for \eqref{eqn:PlVEmbed} to hold. The best known sufficient condition for an arbitrary $T \subset S^{n-1}$ is $m\gtrsim \rho^{-4} \ell_*^2(T)$, established in \cite{OyR15}, while for certain `simple' subsets of the Euclidean sphere (e.g., if $T$ is \green{a subspace}) $m\gtrsim \rho^{-2} \ell_*^2(T)$ is known to be sufficient \cite{OyR15,PlV14}.
\par
\begin{center}
\begin{figure}[htp]
\begin{subfigure}{0.45\textwidth}
\centering
\includegraphics[scale=0.8,trim={210 405 150 206},clip]{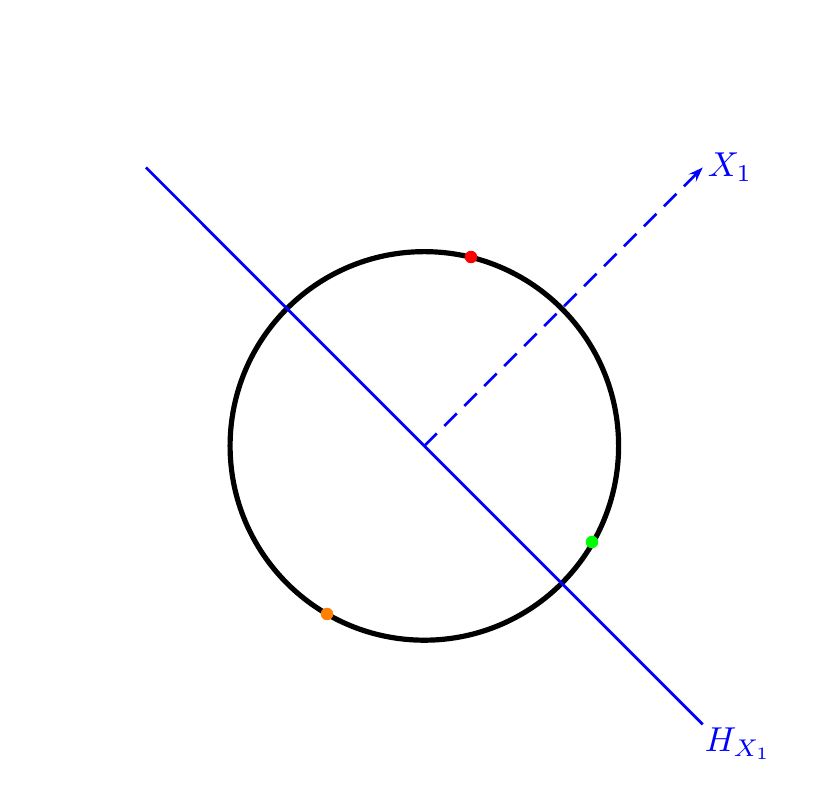}
\end{subfigure}
\begin{subfigure}{0.45\textwidth}
\centering
\includegraphics[scale=0.8,trim={210 405 150 206},clip]{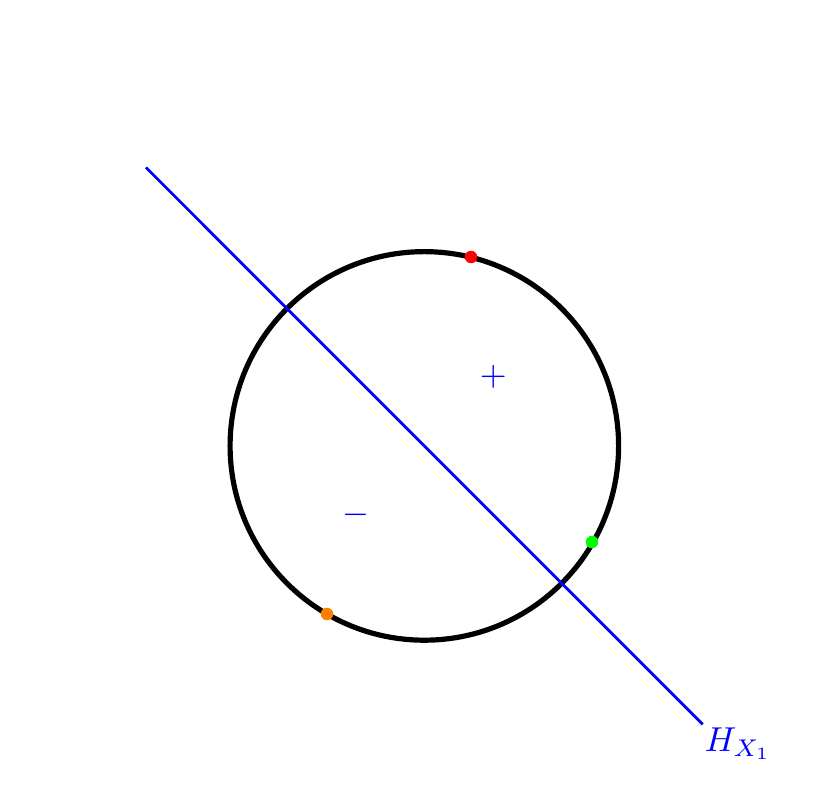}
\end{subfigure}
\caption{Illustration of the hyperplane cut generated by the vector $X_1$ (and shift parameter $0$). The homogeneous hyperplane $H_{X_1}$ divides $\R^n$ into two parts, a ``$+$" and a ``$-$" side. The red and green points are assigned the bit $1$, the orange point is assigned $-1$.}
\label{fig:firstCut}
\end{figure}
\end{center}
\begin{center}
\begin{figure}[htp]
\includegraphics[scale=0.8,trim={230 405 142 206},clip]{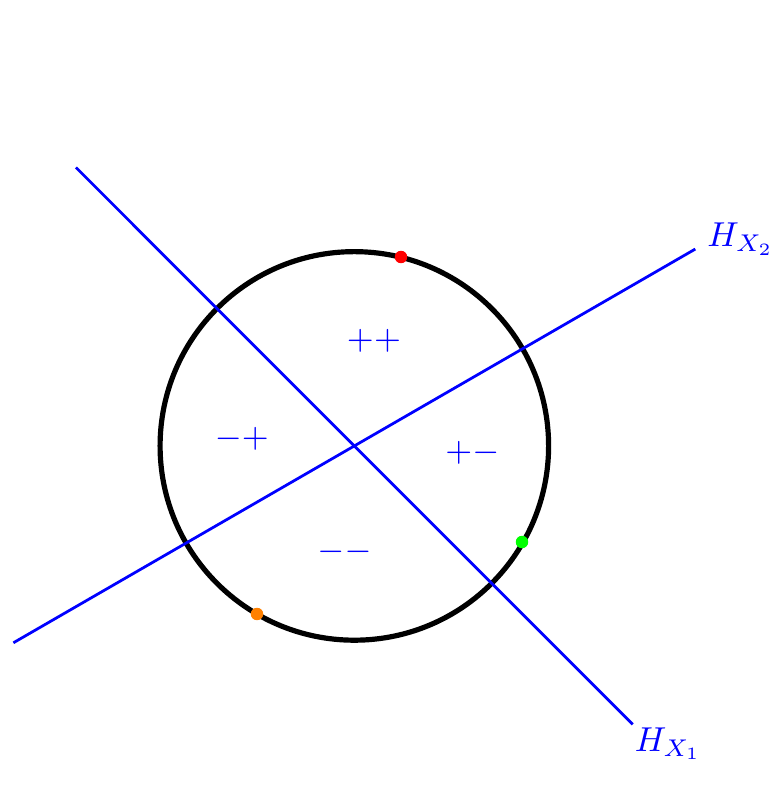}
\caption{The homogeneous hyperplanes $H_{X_1}$ and $H_{X_2}$ divide $\R^n$ into four parts. The red, green, and orange points are assigned the bit sequences $\{1,1\}$, $\{1,-1\}$ and $\{-1,-1\}$, respectively.}
\label{fig:secondCut}
\end{figure}
\end{center}
It is natural to ask whether approximating distances via random tessellations is possible in more general situations, and the obvious cases that come to mind are to consider other distributions for generating the normal vectors (i.e., not Gaussian), and sets $T$ that need not be subsets of $S^{n-1}$. As it happens, these are not only natural but also of extreme importance in signal processing---specifically, when studying signal reconstruction problems from quantized measurements. We describe the connections between the extended version of the random tessellation problem and signal recovery  in detail in Section~\ref{sec:SRHT}.

\par

Unfortunately, it is clear that the two extensions one is interested in are not possible when considering tessellations generated by homogeneous hyperplanes. First of all, it is impossible to separate points lying on a straight line \green{through the origin} using a homogeneous hyperplane. And second, it is easy to find very natural distributions for which \eqref{eqn:PlVEmbed} is false. As an extreme case, if the $X_i$ are i.i.d.\ symmetric Bernoulli random vectors (i.e., the $X_i$'s are selected independently from the uniform distribution on $\{-1,1\}^n$), there are vectors in $S^{n-1}$ that are far apart but still cannot be separated using $H_{X_i}$ --- even if one uses all possible hyperplanes generated by points in $\{-1,1\}^n$.

\green{A possible solution to both problems stems in a phenomenon that appears in the engineering literature: there is extensive experimental evidence that signal recovery from quantized measurements improves substantially if one adds appropriate `noise' to the measurements before quantizing. The operation of adding noise before quantization, which was first proposed in \cite{Rob62}, is called \emph{dithering} (see also the survey \cite{GrN98}). In the context of random tessellations, the geometric interpretation of dithering is adding random parallel shifts to the hyperplanes. As we show in what follows, the addition of such random shifts allows one to address the two problems: random tessellations of arbitrary sets $T$ that are generated by rather general distributions can be used to approximate distances in $T$. Moreover, as an added value, our results explain why dithering is such an effective method in signal recovery problems (see Section \ref{sec:SRHT} for more details).}

\begin{center}
\begin{figure}[htp]
\begin{subfigure}{0.45\textwidth}
\centering
\includegraphics[scale=0.7,trim={210 405 150 220},clip]{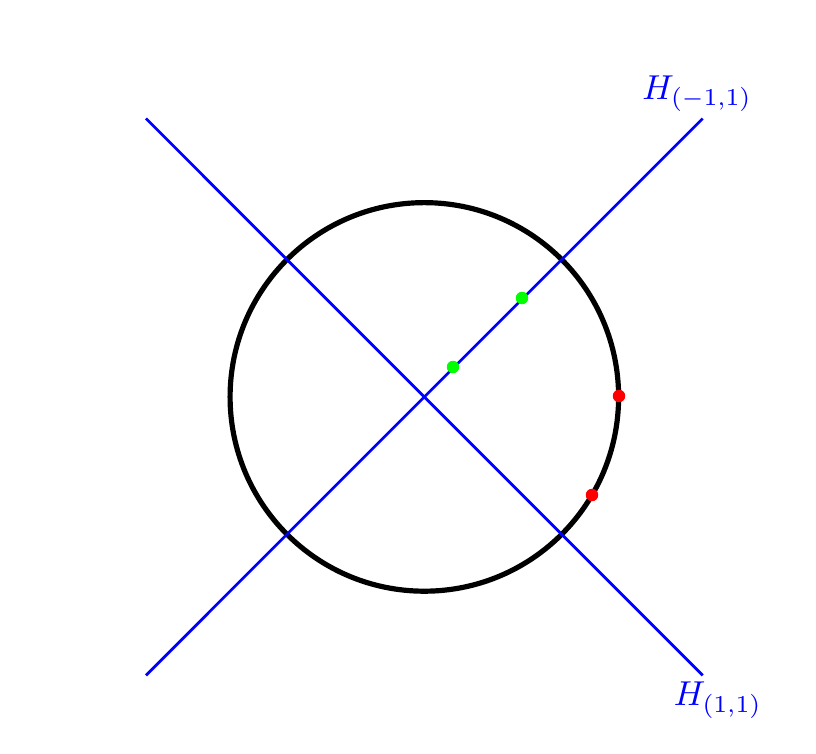}
\end{subfigure}
\begin{subfigure}{0.45\textwidth}
\centering
\includegraphics[scale=0.7,trim={210 405 120 206},clip]{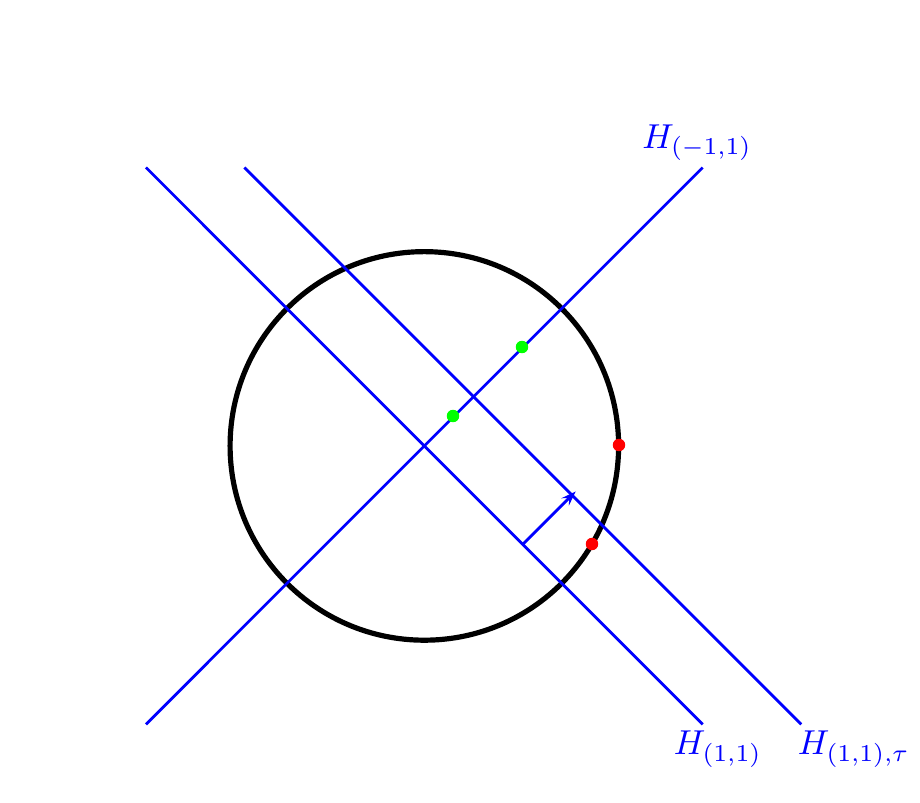}
\end{subfigure}
\caption{Bernoulli vectors in $\R^2$ can only generate two different homogeneous hyperplanes. As a result, there exist two points on the sphere (the examples $e_1$ and $(e_1+\lambda e_2)/\sqrt{1+\la^2}$ for $-1<\lambda<0$ are marked in red) which are far apart, but cannot be separated by a Bernoulli hyperplane. This problem persists in high dimensions. In addition, any two points lying on a straight line through the origin (examples are marked in green) cannot be separated by a homogeneous hyperplane (the latter problem is not specific to the Bernoulli case). Both problems can be solved by using parallel shifts of the hyperplanes.}
\label{fig:BernoulliShift}
\end{figure}
\end{center}
To formulate our results, consider i.i.d.\ shifts $\tau_i$ that are uniformly distributed in $[-\lambda,\lambda]$ for a well chosen $\lambda$, let $T \subset RB_2^n$, the Euclidean ball of radius $R$, set $X$ to be a random vector in $\R^n$ and let $X_1,\ldots,X_m$ be independent copies of $X$ that are also independent of $(\tau_i)_{i=1}^m$. Although the method we introduce can be used in other situations (see in particular Remark~\ref{rem:circulant}), our focus is on two scenarios. The first is an $L$-subgaussian scenario, in which $X$ is isotropic\footnote{Recall than a random vector is isotropic if its covariance matrix is the identity; thus, for every $x \in \R^n$, $\E\inr{X,x}^2 =\|x\|_2^2$.}, symmetric, and $L$-subgaussian, that is, for every $x \in \R^n$ and $p \geq 2$, $\|\inr{X,x}\|_{L^p} \leq L \red{\sqrt{p}}\|\inr{X,x}\|_{L^2}$. The following result is a special case of Theorem~\ref{thm:tess-subgaussian}, and to formulate it we denote by $\co(T)$ the convex hull of the set $T$.
\begin{Theorem} \label{thm:distances-intro}
Set
$$
d(x,y)=\frac{1}{m}|\{i: \sign(\langle X_i,x\rangle + \tau_i) \not=\sign(\langle X_i,y\rangle + \tau_i)\}|.
$$
There exist constants $c_0,\ldots,c_4$ depending only on $L$ such that the following holds. \green{Fix $0<\rho<R$.} If $T \subset R B_2^n$, $\la = c_0 R$ and
$$
m \geq c_1\frac{R \log(eR/\rho)}{\rho^3} \ell_*^2(T),
$$
then with probability at least $1-8\exp(-c_2m\rho/R)$, for any $x,y \in \co(T)$ such that $\|x-y\|_2 \geq \rho$, one has
\begin{equation} \label{eq:isomorphic}
c_3\frac{ \|x-y\|_2}{R} \leq d(x,y) \leq \blue{c_4\sqrt{\log(eR/\rho)}} \cdot \frac{\|x-y\|_2}{R}.
\end{equation}
\end{Theorem}

Theorem~\ref{thm:distances-intro} shows that to approximate Euclidean distances in $T$ it is sufficient to use a number of hyperplanes that is proportional to the squared Gaussian mean width of $T$. The latter quantity is a natural measure of the `intrinsic dimension' of the set. For instance, if $E$ is a $d$-dimensional subspace, then $T=E\cap B_2^n$ has mean width $\ell_*^2(T)\simeq d$. Another example that plays an important role in what follows is $T=\Sigma_{s,n}$, the set of all $s$-sparse vectors in the unit ball, in this case $\ell_*^2(T)\simeq \log {n\choose s}\simeq s\log(en/s)$.

Note that the lower estimate in \eqref{eq:isomorphic} implies that the hyperplanes endow a \emph{$\rho$-uniform tessellation}: any cell of the tessellation of $T$ has diameter at most $\rho$.

\par

In the second scenario we explore \emph{heavy-tailed} random variables: again $X$ is isotropic and symmetric, but in addition we only assume that $X$ satisfies an $L^1$-$L^2$ equivalence:
\begin{equation}
\label{eqn:L1L2equivIntro}
\|\inr{X,x}\|_{L^2} \leq L \|\inr{X,x}\|_{L^1}, \qquad \text{for every } x \in \R^n.
\end{equation}
In the heavy-tailed scenario a different complexity parameter dictates the required number of hyperplanes. For $K\subset\R^n$ we consider
$$
E(K) := \E \sup_{x \in K} \Big|\Big\langle \frac{1}{\sqrt{m}}\sum_{i=1}^m \eps_i X_i,x\Big\rangle\Big|,
$$
where $(\eps_i)_{i\geq 1}$ is a sequence of independent, symmetric $\{-1,1\}$-valued random variables that is independent of $X_1,...,X_m$. If $X_1,\ldots,X_m$ happen to be isotropic, symmetric and subgaussian, then $E(K) \leq c \ell_*(K)$ for an absolute constant $c$.

\begin{Remark}
The fact that $E(K)$ is dominated by the Gaussian mean width of $K$ is one of the features of subgaussian processes and is an outcome of Talagrand's majorizing measures theorem \cite{Tal14}. Finding upper bounds on $E(K)$ when $X$ is not subgaussian is a challenging question that has been studied extensively over the last 30 years or so and which will not be pursued here.
\end{Remark}

Theorem \ref{thm:tessHeavyIntro} is a special case of Theorem~\ref{thm:tess} below. In what follows, given $K \subset \R^n$ and $r>0$ we denote by ${\mathcal N}(K,r)$ the smallest number of Euclidean balls of radius $r$ that are needed to cover $K$.

\begin{Theorem} \label{thm:tessHeavyIntro}
There exist constants $c_0,\ldots,c_4$ that depend only on $L$ for which the following holds. Fix $0<\rho<R$, let $T \subset R B_2^n$ and set $U=\co(T)$. Let $\lambda = \blue{c_0} R$, $r = c_1\rho^2/R$, $U_r=(U-U)\cap rB_2^n$ and assume that
$$
m\geq c_2\left(\left(\frac{R \ E(U_r)}{\rho^2}\right)^2 + \frac{R\log{\mathcal N}(U,r)}{\rho}\right).
$$
Then with probability at least $1-8\exp(-c_3 m (\rho/R)^2)$, for every $x,y \in U$ that satisfy $\|x-y\|_2 \geq \rho$,
\begin{equation}
\label{eqn:tessHeavyIntro}
c_3\frac{ \|x-y\|_2}{R} \leq d(x,y) \leq c_4\frac{R}{\rho} \cdot \frac{\|x-y\|_2}{R}.
\end{equation}
\end{Theorem}

\begin{Remark}
It should come as no surprise that the uniform upper estimate on $d(x,y)$ deteriorates the more `heavy-tailed' the random vector $X$ is, while at the same time the lower bound is universal. It reflects the fact that such lower bounds are due to a small-ball property rather than tail estimates.
This universal behaviour implies that almost regardless of the choice of $X$, if $x$ and $y$ are reasonably `far apart' then their distance is exhibited by the fraction of tessellation hyperplanes that separate the points.
\end{Remark}

The connection between the number of hyperplanes $m$ and the accuracy $\rho$ is less explicit in Theorem~\ref{thm:tessHeavyIntro}, because $E(U_r)$ depends on $m$. And even though the uniform central limit theorem shows that $E(U_r)$ converges to $\ell_*(U_r)$ as $m$ tends to infinity, we are interested in quantitative estimates, which are, in general, nontrivial. Since estimating $E(\red{U_r})$ is not the main focus of this article, we shall not pursue this question any further. We only consider the set $T=\Sigma_{s,n}$ for the sake of illustration. In this case, $U,U_r \subset 4(\sqrt{s}B_1^n \cap B_2^n)=4\co(\Sigma_{s,n})$. By Sudakov's inequality,
$$\log {\mathcal N}(T,r)\leq c_1 \frac{\ell_*^2(U_r)}{r^2} \leq c_2 \frac{R^2s\log(en/s)}{\rho^4}.$$
Moreover, $E(U_r)\leq 4 E(\co(\Sigma_{s,n})) = 4 E(\Sigma_{s,n})$. If the latter parameter can be bounded by a multiple $C\ell_*(\Sigma_{s,n})$ of the Gaussian width, then Theorem~\ref{thm:tessHeavyIntro} shows that \eqref{eqn:tessHeavyIntro} holds if
$$
m=c(L) \frac{R^2s\log(en/s)}{\rho^4},
$$
which is, up to worse dependencies on $R$ and $\rho$, the same scaling as in the subgaussian case. As it happens, this can be guaranteed under very mild assumptions on the vector $X$ by using techniques from \cite{LeM14,Men15}. For instance, $X$ can be any isotropic, unconditional, log-concave vector $X$. More striking is the case where the coordinates of $X$ are independent copies of a mean-zero random variable $\xi$ satisfying, for some $c>0$,
$$
(\E|\xi|^p)^{1/p} \lesssim cp^{\al}, \qquad \text{for all } p\leq \log n.
$$
In this case, the assumption is satisfied with $C\lesssim ce^{2\al-1}$. Thus, one only needs control of the first $\log n$ moments and higher moments may not even exist. This condition is satisfied by a wide variety of extremely heavy-tailed random variables. We refer to \cite[Section V]{DLR16} for proofs and many explicit examples.

\par

\vskip0.3cm

Before we present the proofs of Theorems \ref{thm:distances-intro} and \ref{thm:tessHeavyIntro}, let us explore the connection between random hyperplane tessellations and signal recovery problems. Readers that are solely interested in hyperplane tessellations can safely skip straight to Section~\ref{sec:rand-tess}, where the proofs may be found.

\subsection{Application to one-bit compressed sensing}
\label{sec:SRHT}

One good reason for studying non-Gaussian random hyperplane tessellations of arbitrary sets comes from signal recovery problems involving quantized measurements. By \emph{quantization} we mean converting analog measurements of a signal into a finite number of bits. This essential step is part of any signal processing procedure and allows one to digitally transmit, process, and reconstruct signals. The area of \emph{quantized compressed sensing} investigates how to design a measurement procedure, quantizer, and reconstruction algorithm that together recover low-complexity signals---such as signals that have a sparse representation in a given basis. An efficient system has to be able to reconstruct signals based on a minimal number of measurements, each of which is quantized to the smallest number of bits, and to do so via a computationally efficient reconstruction algorithm. In addition, the system should be reliable: it should be robust to both pre-quantization noise (noise in the analog measurements process) and post-quantization noise (bit corruptions that occur during the quantization process).
\par
Our interest is in the popular \emph{one-bit compressed sensing model}, in which one observes quantized measurements of the form
\begin{equation}
\label{eqn:1-bitModel}
q=\text{sign}(Ax + \noise + \thres),
\end{equation}
where $A\in\R^{m\ti n}$, $m\ll n$, $\text{sign}$ is the sign function applied element-wise, $\noise\in \R^m$ is a vector modelling the noise in the analog measurement process and $\thres\in \R^m$ is a (possibly random) vector consisting of quantization thresholds. We restrict ourselves to memoryless quantization, meaning that the thresholds are set in a non-adaptive manner. In this case,
the one bit quantizer $\text{sign}(\cdot + \thres)$ can be implemented efficiently in practice, and because of its efficiency it has been very popular in the engineering literature---especially in applications in which analog-to-digital converters represent a significant factor in the energy consumption of the measurement system \green{(see e.g.\ \cite{BoB08,MoH15})}.

\par
In spite of its popularity, there are few rigorous results that show that one-bit compressed sensing is viable: the vast majority of the mathematical literature \green{(see e.g.\ \cite{BFN17,JLB13,KSW16,PlV13lin,PlV13})} has focused on the special case where $A$ is a standard Gaussian matrix, and the practical relevance of such results is limited---Gaussian matrices cannot be realized in a real-world measurement setup. As an additional difficulty, it is well known that one-bit compressed sensing may perform poorly outside the Gaussian setup. In fact, it can very easily \emph{fail}, even if the measurement matrix is known to perform optimally in `unquantized' compressed sensing. For example, if the threshold vector $\thres=0$, there are $2$-sparse vectors that cannot be distinguished based on their one-bit Bernoulli measurements (see Figure~\ref{fig:BernoulliShift}).

\par

As an application of the new hyperplane tessellation results described in the previous section, we show that one-bit compressed sensing can actually perform well in scenarios that are far more general than the Gaussian setting. What makes all the difference is the rather striking effect that dithering (that is, adding well-designed `noise' to the measurements before quantizing) has on the one-bit quantizer. Indeed, we show that thanks to dithering, accurate recovery from one-bit measurements is possible even if the measurement vectors are drawn from a heavy-tailed distribution. Moreover, the recovery results we establish are robust to both adversarial and potentially heavy-tailed stochastic noise on the analog measurements, as well as to adversarial bit corruptions that may occur during quantization. In what follows we explain why dithering has such an effect: the geometric interpretation of dithering leads to random tessellations that can be used to approximate distances between signals and the ability to approximate distances has a crucial impact on the performance of recovery procedures.

\par

To understand the connection between hyperplane tessellations and signal recovery from one-bit quantized measurements, let us first assume that no bit corruptions occur in the quantization process; and that there is no pre-quantization noise ($\noise=0$). In this case, we observe $q=\text{sign}(Ax + \thres)$. If we let $X_1,\ldots,X_m$ denote the rows of $A$ and $\tau_1,\ldots,\tau_m$ the entries of $\thres$, then $q$ exactly encodes the cell of the hyperplane tessellation in which the signal $x$ is located. A popular strategy to recover $x$ is to search for a vector $x^{\#}\in T$ that is \emph{quantization consistent}, i.e., $q=\sign(Ax^{\#}+\thres)$. For instance, if $T=\Sigma_{s,n}$, the set of all $s$-sparse vectors in the unit ball, then we can find such a vector by solving
\begin{equation} \label{eq:prog2}
\min_{z \in \R^n} \|z\|_0 \qquad \text{s.t.} \qquad q = \sign(Az+\thres), \  \|z\|_2\leq \red{1}.
\end{equation}
Geometrically, a quantization consistent vector is simply a vector lying in the same cell as $x$. We can ensure that $\|x^{\#}-x\|_2\leq \rho$ by showing that $\|x-y\|_2\leq \rho$ for \emph{any} $y\in T$ located in the same cell as $x$. Since we have no further information on the identity of the cell in which $x$ is located, one has to ensure that any pair of points in $T$ located in the same cell are at distance at most $\rho$ from each other, i.e., the hyperplanes $H_{X_i,\tau_i}$ must induce a \emph{$\rho$-uniform tessellation} of $T$. Phrased differently, if $x,y\in T$ are at distance at least $\rho$, then that fact must be exhibited by the hyperplanes $H_{X_i,\tau_i}$: at least one of the hyperplanes must separate $x$ and $y$. Thus, given a $\rho$-uniform tessellation of $T$, one can uniformly recover signals from $T$ using only $\sign(Ax+\thres)$ as data. Moreover, the reverse direction is clearly true: the degree of accuracy in uniform recovery results in $T$ is determined by the largest diameter (in $T$) of a cell of the tessellation endowed by the hyperplanes $H_{X_i,\tau_i}$.

\par

Unfortunately, even if $(H_{X_i,\tau_i})_{i=1}^m$ induces a uniform tessellation of $T$ there is still the question of pre- and post-quantization noise one has to contend with. To understand the effect of post-quantization noise (i.e., bit corruptions that occur during quantization), assume that one observes a corrupted sequence of bits $q_{\text{corr}} \in \{-1,1\}^m$, where the $i$-th bit being corrupted means that instead of receiving $q_i=\sign(\langle X_i,x\rangle +\tau_i)$ from the quantizer, one observes $(q_{\text{corr}})_i=-\sign(\langle X_i,x\rangle +\tau_i)$; thus, one \red{is led to believe} that $x$ is on the `wrong side' of the $i$-th hyperplane $H_{X_i,\tau_i}$. As a consequence, recovery methods that search for a quantization consistent vector can easily fail even if a single bit is corrupted. For instance, the program \eqref{eq:prog2} (with $q$ replaced by $q_{\text{corr}}$) will in the best case scenario search for a vector in the wrong cell of the tessellation, and in the worse case the corrupted bit may cause a conflict and there will be no sparse vector $z$ satisfying $q_{\text{corr}}=\sign(Az+\thres)$ (see Figure~\ref{fig:bitCorr} for an illustration).

\begin{center}
\begin{figure}[htp]
\includegraphics[scale=0.8,trim={125 550 0 136},clip]{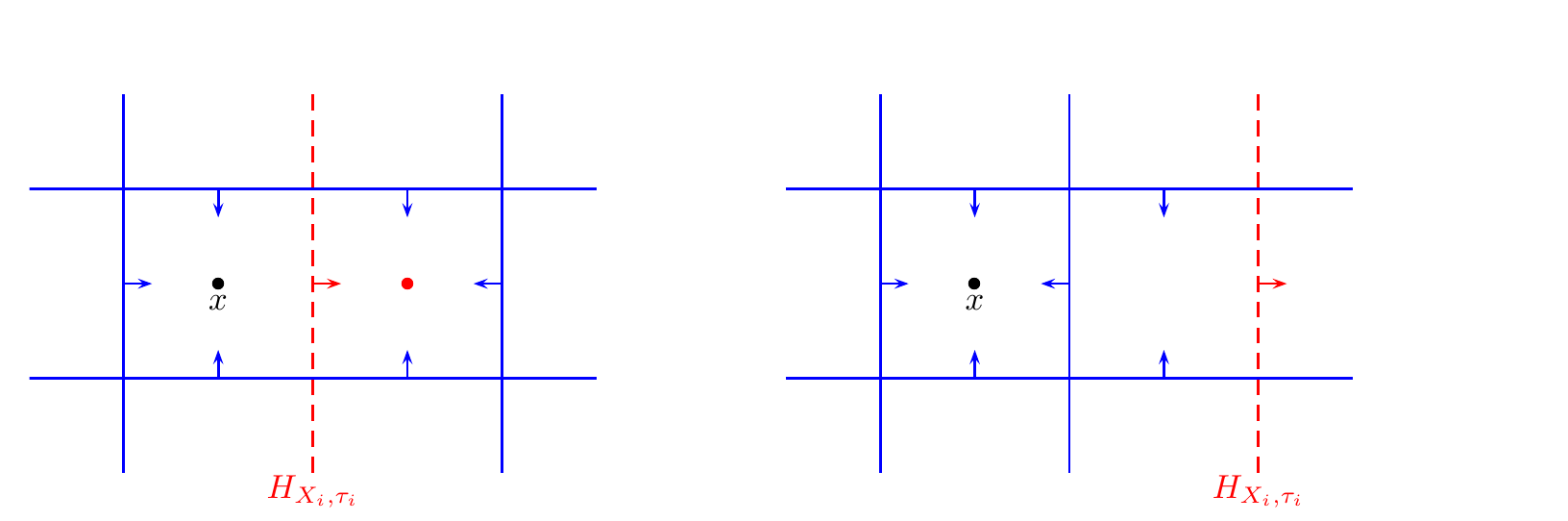}
\caption{The effect of a bit corruption associated with the dashed, red hyperplane $H_{X_i,\tau_i}$. Either the bit corruption leads the program \eqref{eq:prog2} (with $q$ replaced by $q_{\text{corr}}$) to search in the wrong cell of the tessellation marked by the red dot (the picture on the l.h.s.) or causes the program to be infeasible (the picture on the r.h.s.).}
\label{fig:bitCorr}
\end{figure}
\end{center}
The effect of pre-quantization noise (i.e., noise in the analog measurement process) is equally problematic: noise simply causes a parallel shift of the hyperplane $H_{X_i,\tau_i}$, and one has no control over the size of this `noise-induced' shift. Again, the recovery program \eqref{eq:prog2} \green{(with $q=\sign(Ax+\noise+\thres)$)} can easily fail if pre-quantization noise is present (see Figure~\ref{fig:bitCorrNoise}).
\par
\begin{center}
\begin{figure}[htp]
\includegraphics[scale=0.8,trim={125 550 0 136},clip]{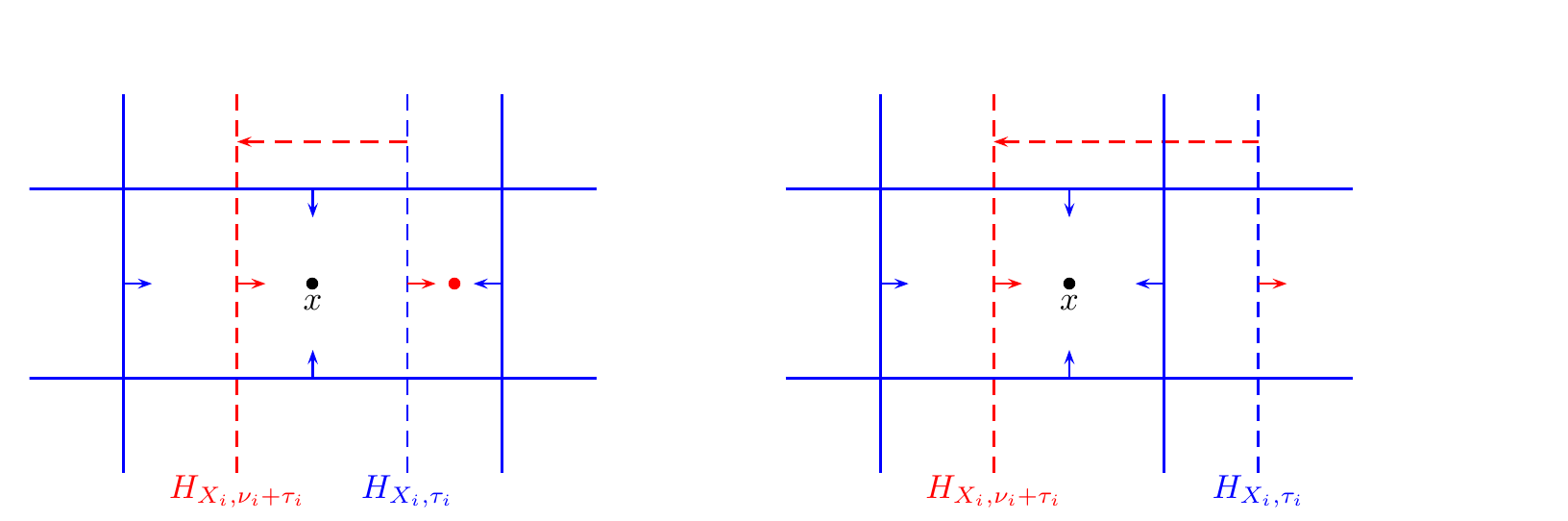}
\caption{The effect of a noise-induced parallel shift of the dashed, blue hyperplane $H_{X_i,\tau_i}$ onto the dashed, red hyperplane \green{$H_{X_i,\nu_i+\tau_i}$}. The program \eqref{eq:prog2} (with $q=\sign(Ax+\noise+\thres)$) searches for a vector $z$ with $\sign(\langle X_i,z\rangle+\tau_i)=\sign(\langle X_i,z\rangle+\nu_i+\tau_i)$. This means that the program incorrectly searches for a solution located to the right of the dashed, blue hyperplane $H_{X_i,\tau_i}$; as a consequence, a solution is found in the wrong cell of the tessellation marked by the red dot (the picture on the l.h.s) or it can even happen that no feasible point exists (the picture on the r.h.s).}
\label{fig:bitCorrNoise}
\end{figure}
\end{center}

One possible way of overcoming this `infeasibility problem' due to noise is by designing a recovery program that is stable: its output does not change by much even if some of the given bits are misleading. For example, one may try search for a vector $z\in T$ whose uncorrupted quantized measurements $\sign(Az+\noise+\thres)$ are closest to the observed corrupted vector $q_{\text{corr}}$. However, since one does not have access to $\noise$, one can only try to match its proxy $\sign(Az+\thres)$ to $q_{\text{corr}}$, i.e. to solve
\begin{equation}
\label{eqn:bitMatchIntro}
\min_{z\in \R^n} d_H(q_{\text{corr}},\sign(Az+\thres)) \qquad {\text{s.t.}} \qquad z\in T,
\end{equation}
\green{where $d_H$ denotes the Hamming distance.} In the context of sparse recovery, the latter program is
\begin{equation}
\label{eqn:bitMatch}
\min_{z\in \R^n} d_H(q_{\text{corr}},\sign(Az+\thres)) \qquad {\text{s.t.}} \qquad \red{\|z\|_0\leq s}, \ \|z\|_2\leq \red{1}.
\end{equation}
To ensure that \red{\eqref{eqn:bitMatchIntro}} yields an accurate reconstruction, the uniform tessellation has to be finer than in the corruption-free case: even if some signs are `flipped', the distance between points in the resulting cell and points in the true one should still be small. Our hyperplane tessellation results ensure this: for any $x,y \in T$ that are at least $\rho$-separated there are \emph{many} hyperplanes that separate the two points---of the order of $\|x-y\|_2 m$. Thus, even after corrupting $\simeq \rho m$ bits one may still detect that $x$ and $y$ are `far away' from one another.

\par

Finally, although \eqref{eqn:bitMatchIntro} can guarantee robust signal recovery, there are no guarantees that it can be solved efficiently. In addition, since \eqref{eqn:bitMatchIntro} matches $\sign(Az+\thres)$, rather than $\sign(Az+\noise+\thres)$, to $q_{\text{corr}}$, it is still quite sensitive to pre-quantization noise. Both problems can be mended by convexification. Indeed, observe that
$$
d_H(q_{\text{corr}},\sign(Az+\noise+\thres)) = \frac{1}{2}\sum_{i=1}^m (1-(q_{\text{corr}})_i\sign(\langle X_i,z\rangle+\nu_i+\tau_i)).
$$
One may relax this objective function by replacing $\sign(\langle X_i,z\rangle+\nu_i+\tau_i)$ by $\langle X_i,z\rangle+\nu_i+\tau_i$ and \red{relax the constraint $z\in T$ to $z\in \co(T)$} leading to the convex program
\begin{equation*}
\min_{z\in \R^n} \frac{1}{2}\sum_{i=1}^m (1-(q_{\text{corr}})_i(\langle X_i,z\rangle+\nu_i+\tau_i)) \qquad \text{s.t.} \qquad z\in \co(T).
\end{equation*}
An equivalent formulation of this program, which only requires the known data $q_{\text{corr}}$ and $A$, is
\begin{equation}
\label{eqn:maxCorPV}
\max_{z\in \R^n} \frac{1}{m}\langle q_{\text{corr}},Az\rangle \qquad \text{s.t.} \qquad \red{z\in \co(T)}.
\end{equation}
This program was proposed in \cite{PlV13} and in what follows we explore a regularized version of \eqref{eqn:maxCorPV}: for $\lambda>0$ we consider
\begin{equation}
\label{eqn:convTintro}
\max_{z\in \R^n} \frac{1}{m}\langle q_{\text{corr}},Az\rangle - \frac{1}{2\la}\|z\|_2^2 \qquad {\text{s.t.}} \qquad z\in \co(T);
\end{equation}
in the context of sparse recovery, this corresponds to the tractable program
\begin{equation*}
\max_{z\in \R^n} \frac{1}{m}\langle q_{\text{corr}},Az\rangle - \frac{1}{2\la}\|z\|_2^2 \qquad \text{s.t.} \qquad \|z\|_1\leq \red{\sqrt{s}}, \ \|z\|_2\leq \red{1}.
\end{equation*}
Let us formulate our main signal recovery results, which are direct outcomes of the results on random tessellations.

Fix a target reconstruction error $\rho$, recall that the quantization thresholds $\tau_i$ are i.i.d.\ uniformly distributed in $[-\la,\la]$, assume that the entries $\nu_i$ of $\noise$ are i.i.d.\ copies of a random variable $\nu$ and that at most $\beta m$ of the bits are arbitrarily corrupted during quantization, i.e., the observed corrupted vector $q_{\text{corr}}$ satisfies $d_H(q_{\text{corr}},q)\leq \beta m$. The adversarial component of the pre-quantization noise $\nu$ is $|\E \nu |$, $\sigma^2$ is its variance and $\|\nu\|_{L^2}$ is its $L^2$ norm. We write $T_r = (T-T) \cap r B_2^n$ for any $r>0$.

\par

Our first recovery result concerns the recovery program \eqref{eqn:bitMatchIntro} in the $L$-subgaussian scenario, in which the rows $X_i$ of $A$ are i.i.d.\ copies of a symmetric, isotropic, $L$-subgaussian vector $X$. In addition, we assume that $\nu$ is $L$-subgaussian: for every $p \geq 2$, $\|\nu\|_{L^p} \leq L \sqrt{p}\|\nu\|_{L^2}$.

\begin{Theorem} \label{thm:mainIntroNC-sub}
There exist constants $c_0,\ldots,c_4>0$ depending only on $L$ such that the following holds. Let $T \subset RB_2^n$,  set $\lambda \geq c_0 (R + \|\nu\|_{L^2}) + \rho$ and put $r=c_1 \rho/\sqrt{\log(e\lambda/\rho)}$. Assume that
$$
\blue{m \geq c_2 \lambda \left( \frac{\ell_*^2(T_r)}{\rho^3} + \frac{\log {\mathcal N}(T, r)}{\rho}\right),}
$$
and that $|\E \nu| \leq c_3\rho$, $\sigma \leq c_3 \rho/\sqrt{\log(e \lambda/\rho)}$ and $\beta \leq c_3 \rho/\lambda$.

Then with probability at least $1-10\exp(-c_4 m \rho/\lambda)$, for every $x \in T$,
any solution $x^{\#}$ of \eqref{eqn:bitMatchIntro} satisfies $\|x^{\#}-x\|_2 \leq \rho$.
\end{Theorem}

To put Theorem~\ref{thm:mainIntroNC-sub} in some context, consider an arbitrary $T \subset B_2^n$ and assume $\|\nu\|_{L^2} \leq 1$, \red{so that} $\lambda$ is a constant that depends only on $L$. By Sudakov's inequality,
\begin{equation}
\label{eqn:sudakovAppl}
\log {\mathcal N}(T, r) \leq c\frac{\ell_*^2(T)}{r^2} \leq c(L) \frac{\log(e/\rho)}{\rho^2}  \ell_*^2(T),
\end{equation}
\red{and trivially} $\ell_*(T_r) \leq \ell_*(T)$, which means that a sample size of
$$
m =c^\prime(L) \frac{\log(e/\rho)}{\rho^3} \ell_*^2(T)
$$
suffices for recovery. In the special case of $T=\Si_{s,n}$, the subset of $B_2^n$ consisting of $s$-sparse vectors, a much better estimate is possible. Indeed, it is standard to verify that there is an absolute constant $c$ such that for any $1 \leq s \leq n$,
\begin{equation}
\label{eqn:covNumEstSparse}
\ell_*(\Si_{s,n}) \simeq \sqrt{s\log(en/s)} \ \ \ {\rm and} \ \ \ \log {\mathcal N}(\Si_{s,n},r) \leq cs\log\left(\frac{en}{sr}\right).
\end{equation}
Moreover, since $(\Si_{s,n} - \Si_{s,n}) \cap r B_2^n \subset r \Si_{2s,n}$ it follows that
$$
\ell_*(T_r) \leq c r \sqrt{s \log(en/s)} = c(L) \frac{\rho}{ \sqrt{\log(e/\rho)}} \cdot \sqrt{s \log(en/s)},
$$
implying that a sample size of
\begin{equation}
\label{eqn:optScalSparse}
m =c'(L) \rho^{-1} s\log\left(\frac{en}{s\rho}\right)
\end{equation}
guarantees that with high probability one can recover any $s$-sparse vector in $B_2^n$ with accuracy $\rho$ via \eqref{eqn:bitMatchIntro}.
\par

\vskip0.4cm

In the heavy-tailed scenario, where we only assume that $X$ is isotropic, symmetric, and satisfies the $L^1$-$L^2$ equivalence \eqref{eqn:L1L2equivIntro}, we obtain the following result. In this setting, we assume that $\nu$ has finite variance $\si^2$ and satisfies an $L^1$-$L^2$ equivalence.
\begin{Theorem} \label{thm:mainIntroNC-heavy}
There exist constants $c_0,\ldots,c_4>0$ depending only on $L$ such that the following holds. Assume that $T \subset R B_2^n$. Let $\lambda \geq c_0 (R + \|\nu\|_{L^2})+\rho$, \blue{set $r=c_1\rho^2/\lambda$, and suppose that $m$ satisfies}
\begin{equation}
\label{eqn:mainIntroNC-heavyCond}
\blue{m \geq c_2 \left(\left(\frac{\lambda E(T_r)}{\rho^2}\right)^2 + \lambda \frac{\log {\mathcal N}(T,r)}{\rho}\right).}
\end{equation}
Assume further that $|\E \nu| \leq c_3\rho$, $\sigma \leq c_3 \rho^{3/2}/\sqrt{\lambda}$ and $\beta \leq c_3 \rho/\lambda$.

Then with probability at least $1-10\exp(-c_4 m (\rho/\lambda)^2)$, for every $x \in T$,
any solution $x^{\#}$ of \eqref{eqn:bitMatchIntro} satisfies $\|x^{\#}-x\|_2 \leq \rho$.
\end{Theorem}
To illustrate this result, assume $\|\nu\|_{L^2} \leq 1$ and consider $T=\Si_{s,n}$, so that $\lambda$ is a constant depending only on $L$. Since $T_r\subset r \Si_{2s,n}$, the first term in \eqref{eqn:mainIntroNC-heavyCond} is bounded by $E^2(\Si_{2s,n})$. The latter can be bounded by $\ell_*^2(\Si_{2s,n})$ under the assumptions on $X$ mentioned after Theorem~\ref{thm:tessHeavyIntro}. Taking into account \eqref{eqn:covNumEstSparse}, it follows that even for these heavy-tailed vectors the sample size \eqref{eqn:optScalSparse} is sufficient for recovery.

\par

Let us compare Theorems~\ref{thm:mainIntroNC-sub} and \ref{thm:mainIntroNC-heavy} to existing work. As was mentioned previously, almost all the signal reconstruction results in (memoryless) one-bit compressed sensing concern standard Gaussian measurement matrices, see e.g.\ \cite{DJR17} for an overview. The most closely related work is \cite{JLB13}, which concerns the situation when there is no dithering ($\thres=0$). Recall that in that case it is only possible to recover signals located on the unit sphere. It was shown in \cite[Theorem 2]{JLB13} that if $A\in \R^{m\times n}$ is standard Gaussian and $m\gtrsim \rho^{-1}s\log(n/\rho)$ then, with high probability, any $s$-sparse $x,x' \in S^{n-1}$ for which $\sign(Ax)=\sign(Ax')$ satisfy $\|x-x'\|_2\leq \rho$. In particular, one can approximate $x$ with accuracy $\rho$ by solving the non-convex program
$$
\min_{z\in \R^n} \|z\|_0 \qquad \text{s.t.} \qquad \sign(Ax)=\sign(Az), \ \|z\|_2=1.
$$
In comparison, Theorem~\ref{thm:mainIntroNC-sub} shows that the same result holds in the subgaussian scenario --- and at the same time extends it to sparse vectors in the unit ball and makes it robust to pre- and post-quantization noise. Clearly, such a generalization is possible thanks to the effect of dithering. \green{Remarkably, Theorem~\ref{thm:mainIntroNC-heavy} shows that this result can be extended further to a large class of heavy-tailed measurements. This is the first known recovery result involving quantized heavy-tailed measurements}.

\par

In \cite{BFN17,KSW16} the authors study sparse recovery with Gaussian measurements and introduce standard Gaussian dithering to derive recovery results for sparse vectors in the unit ball. The idea behind these results is to use a `lifting trick': for instance, in \cite{BFN17} one interprets the dithered measurements $\sign(Ax+\tau)$ as $\sign([A \ \tau] [x,1]/\|[x,1]\|_2)$, where $[A \ \tau]$ is obtained by appending $\tau$ to $A$ as an additional column. Since $[A \ \tau]$ is a standard Gaussian again, recovery methods for sparse vectors with unit norm can be used to find an approximation of $[x,1]/\|[x,1]\|_2$ of the form $[x^{\#},1]/\|[x^{\#},1]\|_2$. Afterwards, one can bound $\|x-x^{\#}\|_2$ by the distance between the latter two vectors. Since this lifting argument is based on a reduction to the one-bit compressed sensing with zero thresholds model, it `imports' the strong limitations of that model; in particular, it cannot be used to derive recovery results for a general class of non-Gaussian measurements. In addition, since the recovery methods in \cite{BFN17,KSW16} rely on enforcing quantization consistency, they are not robust to post-quantization noise. In contrast, thanks to our new geometric understanding of the effect of dithering, we find robust recovery results for non-Gaussian measurements matrices and general signal sets.

\vskip0.4cm

Finally, let us present our main recovery result for \eqref{eqn:convTintro}. We only analyze this recovery program in the $L$-subgaussian scenario. We again assume that $\nu$ is $L$-subgaussian with variance $\sigma^2$. We assume for the sake of simplicity that $\E \nu =0$ and denote $U=\conv(T)$ and $U_\rho = (U-U) \cap \rho B_2^n$.
\begin{Theorem} \label{thm:convex-iid-rows}
There exist constants $c_0,\ldots,c_4$ that depend only on $L$ for which the following holds. Let $T \subset R B_2^n$, fix $\rho>0$, set
$$
\lambda \geq c_0 (\sigma + R) \sqrt{\log(c_0/\rho)}
$$
and let $r=c_1 \rho/\log(e\lambda/\rho)$. If $m$ and $\beta$ satisfy
$$
m \geq c_2 \left(\left(\frac{\lambda \ell_*(U_\rho)}{\rho^2}\right)^2 + \red{\lambda^2} \frac{\log {\mathcal N}(T,r)}{\red{\rho^2}}\right), \qquad \beta\sqrt{\log(e/\beta)}=c_3\frac{\rho}{\lambda},
$$
then, with probability at least $1-8\exp(-c_4 m\rho^2/\lambda^2)$, for any $x \in T$ the solution $x^{\#}$ of \eqref{eqn:convTintro} satisfies $\|x^{\#}-x\|_2 \leq\rho$.
\end{Theorem}
As an example, let $T=\sqrt{s}B_1^n \cap B_2^n$, the set of approximately $s$-sparse vectors in the Euclidean unit ball, and assume that $\sigma\leq 1$. Observe that $T=U$ and that one may set $\lambda=c_0(L)\sqrt{\log(e/\rho)}$. Also, for $0<\rho \leq 1$, $U_\rho \subset 2(\sqrt{s}B_1^n \cap \rho B_2^n)$, and it is standard to verify that $\ell_*(U_\rho) \simeq \sqrt{s \max\{\log(en\rho^2/s),1\}}$. \red{Taking the estimate \eqref{eqn:sudakovAppl} for $\log {\mathcal N}(T,r)$ into account, it is evident that if}
$$
m=c(L)\frac{s\log(en/s)\red{\log^3(e/\rho)}}{\rho^4}
$$
then with high probability one may recover any $x \in T$ using the convex recovery procedure \eqref{eqn:convTintro}, even in the presence of pre- and post-quantization noise.\par
In the context of Gaussian measurement matrices, Theorem~\ref{thm:convex-iid-rows} improves upon work of Plan and Vershynin \cite{PlV13}, who considered the situation when there is no dithering ($\thres=0$). They introduced the convex program \eqref{eqn:maxCorPV} and proved recovery results \red{for signal sets $T\subset S^{n-1}$} of two different flavours. In a non-uniform recovery setting\footnote{\red{In the uniform recovery setting one attains a high probability event on which recovery is possible for all $x \in T$, whereas in non-uniform recovery the event depends on the signal $x\in T$.}} they showed that $m\gtrsim \rho^{-4}\red{\ell_*^2(T)}$ measurements suffice to reconstruct a fixed signal, even if pre-quantization noise is present and quantization bits are randomly flipped with a probability that is allowed to be arbitrarily close to $1/2$. In the uniform recovery setting, they showed that if $m\gtrsim \rho^{-12}\red{\ell_*^2(T)}$, one can achieve a reconstruction error $\rho$ even if a fraction $\beta=\rho^2$ of the received bits are corrupted in an adversarial manner while quantizing. Theorem~\ref{thm:convex-iid-rows} extends the latter result to subgaussian measurements with a better condition on $m$ and $\beta$, and at the same time incorporates pre-quantization noise and allows the reconstruction of signals that need not be located on the unit sphere.\par
When the measurements are not standard Gaussian, there are very few reconstruction results available. The work \cite{ALP14} generalized the non-uniform recovery results from \cite{PlV13} to subgaussian measurements under additional restrictions. For a fixed $x\in T$ with $T\subset S^{n-1}$ they showed that $m\gtrsim \rho^{-4}\ell_*^2(T)$ suffice to reconstruct $x$ up to error $\rho$ via \eqref{eqn:maxCorPV} provided that either $\|x\|_{\infty}\leq \rho^4$ (meaning that the signal must be sufficiently spread) or the total variation distance between the subgaussian measurements and the standard Gaussian distribution is at most $\rho^{16}$. Theorem~\ref{thm:convex-iid-rows} significantly improves on these results.

\vskip0.4cm

\begin{Remark}
\label{rem:circulant}
At the expense of substantial additional technicalities, the proof strategies developed in this work lead to recovery results for sparse vectors when $A$ is a \blue{random partial circulant matrix} generated by a subgaussian random vector. The latter model occurs in several practical measurement setups, including SAR radar imaging, Fourier optical imaging and channel estimation (see e.g.\ \cite{Rom09} and the references therein). To keep this work accessible to a general audience and clearly expose the main ideas, we choose to defer the additional technical developments needed for the circulant case to a companion work \red{\cite{DiM18b}}.
\end{Remark}

\subsection{Notation}

We use $\|x\|_p$ to denote the $\ell_p$-norm of $x\in \R^n$ and $B_p^n$ denotes the $\ell_p$-unit ball in $\R^n$. For a subgaussian random variable $\xi$ we let
$$\|\xi\|_{\psi_2} := \sup_{p\geq 1}\frac{\|\xi\|_{L^p}}{\sqrt{p}\|\xi\|_{L^2}}.$$ \blue{We use $\mathcal{U}$ to denote the uniform distribution.} For $k\in \N$ we set $[k]=\{1,\ldots,k\}$ and for a set $S$ we let $|S|$ denote its cardinality. $d_H$ is the (unnormalized) Hamming distance on the discrete cube and $\Si_{s,n}=\{x\in \R^n \ : \ \|x\|_0\leq s, \ \|x\|_2\leq 1\}$ is the set of $s$-sparse vectors in the Euclidean unit ball. For $T \subset \R^n$ we set $T_r = (T-T) \cap r B_2^n$  and denote by $\co(T)$ its convex hull. As before, we use $\ell_*(T)$ to denote the Gaussian mean width of $T$ and, for any $r>0$, we denote by ${\mathcal N}(T,r)$ the smallest number of Euclidean balls of radius $r$ that are needed to cover $T$. Finally, $c$ and $C$ denote absolute constants; their value many change from line to line. $c_\alpha$ or $c(\alpha)$ denotes a constant that depends only on the parameter $\alpha$. We write $a\lesssim_{\al} b$ if $a\leq c_{\al} b$, and $a\simeq_{\al} b$ means that both $a\lesssim_{\al} b$ and $a\gtrsim_{\al} b$ hold.

\section{Random tessellations} \label{sec:rand-tess}

This section is devoted to the proof of our main tessellation results, Theorems~\ref{thm:tess} and \ref{thm:tess-subgaussian}, which are generalizations of Theorems~\ref{thm:distances-intro} and \ref{thm:tessHeavyIntro} respectively. Before we formulate the results let us define a mild structural property of a subset of a metric space.
\begin{Definition} \label{def:r-metric-convex}
Let $({\mathcal X},d)$ be a metric space. A set $T \subset {\mathcal X}$ is $(r,\gamma)$-metrically convex in ${\mathcal X}$ if for every $x,y \in T$ there are $z_1,...,z_{\ell} \in {\mathcal X}$ such that
$$
\gamma r \leq d(z_i,z_{i+1}) \leq r \ \ {\rm  and} \ \  \sum_{i=0}^\ell d(z_i,z_{i+1}) \leq \gamma^{-1} d(x,y),
$$
where we set $z_0=x$, $z_{\ell+1} = y$. If ${\mathcal X}=T$, then we say that $T$ is $(r,\gamma)$-metrically convex.
\end{Definition}

The idea behind this notion is straightforward: it implies that controlling `local oscillations' of a function $f$ ensures that it satisfies a Lipschitz condition for long distances. Indeed, assume that $\blue{\sup_{\{w,v \in \mathcal{X} : d(w,v) \leq r\}} |f(w)-f(v)| \leq \kappa}$ and for any $x,y \in T$ that satisfy $d(x,y) \geq 2r$ let $(z_i)_{i=0}^{\ell+1}$ be as in Definition \ref{def:r-metric-convex}. Then
\begin{equation} \label{eq:osc-to-lip}
|f(x)-f(y)| \leq \left|\sum_{i=0}^\ell \left(f(z_i)-f(z_{i+1})\right)\right| \leq \kappa (\ell+1)  \leq \frac{\kappa}{\gamma r} \sum_{i=0}^\ell d(z_i,z_{i+1}) \leq \frac{\kappa}{\gamma^2 r} d(x,y).
\end{equation}
Therefore, $f$ satisfies a Lipschitz condition for long distances with constant $\kappa/\gamma^2 r$.
\par
Observe that if $T$ is \red{a convex} subset of a normed space then it is $(r,1)$-metrically convex for any $r>0$; also, every subset of a normed space is $(r,1)$-metrically convex in its convex hull. Finally, $\Sigma_{s,n}$ is $(r,\gamma)$-metrically convex in $\Si_{2s,n}$ for an absolute constant $\gamma$. We omit the straightforward proofs of these claims.\par

Let us first state our main result in the heavy-tailed scenario. We consider a random vector $X$ that is isotropic, symmetric, and satisfies an $L^1$-$L^2$ norm equivalence: i.e, that for every $t \in \R^n$,
\begin{equation} \label{eqn:L1L2ass}
\|t\|_2= \|\inr{X,t}\|_{L^2} \leq L \|\inr{X,t}\|_{L^1}.
\end{equation}
\begin{Theorem} \label{thm:tess}
There exist constants $c_0,\ldots,\red{c_4}$ that depend only on $L$ for which the following holds. Let $T \subset R B_2^n$ and set $\lambda \geq \blue{c_0} R$. Suppose that $0<r <\rho < \lambda$ satisfy $r \leq \blue{c_1}\rho^2/\lambda$ and assume that
$$
\log{\mathcal N}(T,r) \leq c_2 \frac{m \rho}{\lambda}, \ \ \ \ {\rm and} \ \ \ \ E(T_r) \leq c_2 \frac{\rho^2}{\lambda} \sqrt{m}.
$$
Then with probability at least $1-8\exp(-c_3 m (\rho/\lambda)^2)$, for every $x,y \in T$ that satisfy $\|x-y\|_2 \geq \rho$,
$$
|\{i: \sign(\inr{X_i,x} + \tau_i) \not=\sign(\inr{X_i,y} + \tau_i)\}| \geq c_4m\frac{\|x-y\|_2}{\lambda}.
$$
Moreover, if $T$ is $(r,\gamma)$-metrically convex then on the same event, if $\|x-y\|_2 \geq 2r$,
$$
|\{i: \sign(\inr{X_i,x} + \tau_i) \not=\sign(\inr{X_i,y} + \tau_i)\}| \leq \frac{c_5\lambda}{\rho\gamma^2} \cdot m\frac{\|x-y\|_2}{\lambda}.
$$
\end{Theorem}

\vskip0.4cm
\noindent{\bf Proof of Theorem \ref{thm:tessHeavyIntro}.}
Apply Theorem \ref{thm:tess} to the set $U = \co(T)$, which is $(r,1)$ metrically convex for any $r>0$, and for the parameters $\lambda=c_0R$ and $r=c_1\rho^2/R$. With these choices Theorem \ref{thm:tessHeavyIntro} follows immediately.
\endproof

When $X$ is $L$-subgaussian one may establish a sharper result.
\begin{Theorem} \label{thm:tess-subgaussian}
There exist constants $c_0,\ldots,c_5$ that depend only on $L$ for which the following holds. Let $T \subset R B_2^n$, set $\lambda \geq c_0 R$ and consider an isotropic, symmetric, $L$-subgaussian random vector $X$. Let $m$ and $0<r <\rho< \lambda$ satisfy
$$
\rho \geq c_1 r \sqrt{\log(e\lambda/\rho)},
$$
and
$$
m \geq c_2 \max\left\{\frac{\lambda}{\rho} \log{\mathcal N}(T,r), \ \lambda \frac{\ell_*^2(T_r)}{\rho^3} \right\}.
$$
Then with probability at least $1-8\exp(-c_3 m \rho/\lambda)$, for every $x,y \in T$ such that $\|x-y\|_2 \geq \rho$, one has
$$
|\{i: \sign(\inr{X_i,x} + \tau_i) \not=\sign(\inr{X_i,y} + \tau_i)\}| \geq c_4m\frac{\|x-y\|_2}{\lambda}.
$$
Moreover, if $T$ is $(r,\gamma)$-metrically convex then on the same event, if $\|x-y\|_2 \geq 2r$,
$$
|\{i: \sign(\inr{X_i,x} + \tau_i) \not=\sign(\inr{X_i,y} + \tau_i)\}| \leq \blue{\frac{c_5\sqrt{\log(e\lambda/\rho)}}{\gamma^2}} \cdot m\frac{\|x-y\|_2}{\lambda}.
$$
\end{Theorem}

\vskip0.4cm
\noindent{\bf Proof of Theorem \ref{thm:distances-intro}.}
Theorem \ref{thm:distances-intro} is an immediate outcome of Theorem \ref{thm:tess-subgaussian} for $U=\co(T)$. Indeed, $\co(T)$ is $(r,1)$ metrically convex for any $r>0$, $\ell_*(U_r) \leq \ell_*(T)$, and by Sudakov's inequality,
$\log {\mathcal N}(U,r) \leq c \ell_*^2(T)/r^2$. The claim follows by setting $r=c\rho/\sqrt{\log(e\lambda/\rho)}$ and $\lambda=c^\prime R$ for suitable absolute constants $c$ and $c^\prime$.
\endproof

\vskip0.4cm
In the context of tessellations, Theorem~\ref{thm:tess} and the first part of Theorem~\ref{thm:tess-subgaussian} improve the estimate from \eqref{eqn:PlVEmbed} in several ways: firstly, Theorem~\ref{thm:tess} holds for a \red{very} general collection of random vectors $X$ \red{-} the vector has to satisfy a small-ball condition rather than being Gaussian. Secondly, both are valid for any subset of $\R^n$ and not just for subsets of the sphere; and, finally, if $X$ happens to be $L$-subgaussian, it yields the best known estimate on the diameter of each `cell' in the random tessellation---even when $X$ is Gaussian and $T$ is a subset of $S^{n-1}$.

\subsection{The heavy-tailed scenario}
A fundamental question that is at the heart of our arguments has to do with stability: given two points $x$ and $y$, how `stable' is the set
$$
\{ i : \sign(\inr{X_i,x}+\tau_i) \not= \sign(\inr{X_i,y}+\tau_i)\}=(*)
$$
to perturbations? If one believes that the cardinality of $(*)$ reflects the distance $\|x-y\|_2$, it stands to reason that if $r$ is significantly smaller than $\|x-y\|_2$ and $\|x-x^\prime\|_2 \leq r$, $\|y-y^\prime\|_2 \leq r$,  then $|\{ i : \sign(\inr{X_i,x^\prime}+\tau_i) \not= \sign(\inr{X_i,y^\prime}+\tau_i)\}|$ should not be very different from $|(*)|$.

Unfortunately, stability is not true in general. If either $x$ or $y$ are `too close' to many of the separating hyperplanes, then even a small shift in either one of them can have a dramatic effect on the signs of $\inr{X_i,\cdot} +\tau_i$ and destroy the separation. Thus, to ensure stability one requires a stronger property than mere separation: points need to be separated by a large margin.

\begin{Definition} \label{def:large-margin}
The hyperplane $H_{X_i,\tau_i}$ \emph{$\theta$-well-separates} $x$ and $y$ if
\begin{itemize}
\item $\sign(\inr{X_i,x}+\tau_i)\neq \sign(\inr{X_i,y}+ \tau_i)$,
\item $|\inr{X_i,x}+\tau_i|\geq \theta\|x-y\|_2$, and
\item $|\inr{X_i,y}+\tau_i|\geq \theta\|x-y\|_2$.
\end{itemize}
Denote by $I_{x,y}(\theta)\subset[m]$ the set of indices for which $H_{X_i,\tau_i}$ $\theta$-well-separates $x$ and $y$.
\end{Definition}

The condition that $|\inr{X_i,x}+\tau_i|, \ |\inr{X_i,x}+\tau_i| \geq \theta\|x-y\|_2$ is precisely what ensures that perturbations of $x$ or $y$ of the order of $\|x-y\|_2$ do not spoil the fact that the hyperplane $H_{X_i,\tau_i}$ separates the two points.

\par

We begin by showing that even in the heavy-tailed scenario and with high probability, $|I_{x,y}(\theta)|$ is proportional to $m \|x-y\|_2$ for any two (fixed) points $x$ and $y$. Let us stress that the high probability estimate is crucial: it will lead to a uniform control on a net of a large cardinality.

\begin{Theorem} \label{thm:individual-separation}
\red{There are constants $c_1,\ldots,c_4$} that depend only on $L$ for which the following holds. Let $x,y \in R B_2^n$ and set $\lambda \geq c_1 R$. With probability at least
$$
1-\red{4}\exp\left(-c_2 m\min\left\{\frac{\|x-y\|_2}{\lambda},1\right\}\right),
$$
$$
|I_{x,y} (c_3)| \geq  c_4 m \frac{\|x-y\|_2}{\lambda}.
$$
\end{Theorem}

The proof of Theorem \ref{thm:individual-separation} requires two preliminary observations. Consider a random variable $\tau$ that satisfies the small ball estimate
\begin{equation} \label{eqn:SBass}
\sup_{u\in \R} \bP(|\tau-u|\leq \eps) \leq C_\tau \eps \qquad \text{for all } \eps\geq 0,
\end{equation}
and let $Z$ be independent of $\tau$. Then clearly
\begin{equation}
\label{eqn:SBindep}
\bP(|Z+\tau|\leq \eps) \leq C_\tau \eps, \qquad \text{for all } \eps\geq 0.
\end{equation}
\red{If} $\tau \sim {\mathcal U}[-\la,\la]$  then \eqref{eqn:SBass} holds for $C_\tau=1/\lambda$. \red{Therefore, by the Chernoff bound,} if $(Z_i)_{i=1}^m$ and $(\tau_i)_{i=1}^m$ are independent copies of $Z$ and $\tau$ respectively, then with probability at least $1-2\exp(-c m \eps/\lambda)$,
\begin{equation} \label{eq:single-upper}
|\{i: |Z_i + \tau_i| \geq \eps\}| \geq \left(1-\frac{2\eps}{\lambda}\right)m.
\end{equation}

The second observation is somewhat more involved. Consider a random variable $\tau$ that satisfies
\begin{equation}
\label{eqn:hypSepMinAssTauL1}
\bP(\alpha<\tau\leq\beta)\geq c_\tau (\beta-\alpha)
\end{equation}
for all $-\la\leq \alpha\leq\beta\leq\la$. Let $Z$ and $W$ be square integrable whose difference satisfies a small-ball condition: there are constants $\kappa$ and $\delta$ such that
$$
\bP(|Z-W| \geq \kappa \|Z-W\|_{L^1} ) \geq \delta.
$$
\begin{Lemma} \label{lemma:Z-W-tau}
There are absolute constants $c_0$ and $c_1$ and constants $c_2,c_3 \simeq c_\tau \kappa \delta$ such that the following holds. Assume that $Z$ and $W$ are independent of $\tau$ and that
$$
\lambda \geq (c_0/\sqrt{\delta}) \max\{\|Z\|_{L^2},\|W\|_{L^2}\}.
$$
If $(\tau_i)_{i=1}^m$, $(Z_i)_{i=1}^m$ and $(W_i)_{i=1}^m$ are independent copies of $\tau$, $Z$ and $W$ respectively, then with probability at least
$$
1-2\exp(-c_1 m \delta)-2\exp(-c_2 m \|Z-W\|_{L^1}),
$$
we have
$$
|\{i : \sign(Z_i+\tau_i) \not = \sign(W_i+\tau_i)\}| \geq c_3 m \|Z-W\|_{L^1}.
$$
\end{Lemma}

\begin{proof}
Set $\theta$ to be named later and observe that $\bP(|Z| \geq \|Z\|_{L^2}/\sqrt{\theta}) \leq \theta$. Hence, with probability at least $1-2\exp(-c_1\theta m)$,
$$
|\{i : |Z_i| \geq   \|Z\|_{L^2}/\sqrt{\theta}\}| \leq 2\theta m,
$$
where $c_1$ is an absolute constant; a similar estimate holds for $(W_i)_{i=1}^m$.

At the same time, recall that $\bP(|Z-W| \geq \kappa \|Z-W\|_{L^1}) \geq \delta$, implying that with probability at least $1-2\exp(-c_2\delta m)$,
$$
|\{i : |Z_i-W_i| \geq   \kappa \|Z-W\|_{L^1}\}| \geq \frac{\delta m}{2}.
$$
Set $\theta =\delta/16$ and let $\lambda \geq 4\max\{\|Z\|_{L^2}/\sqrt{\delta}, \ \|W\|_{L^2}/\sqrt{\delta}\}$. The above shows that there is an event ${\mathcal A}$ of $(Z,W)$-probability at least $1-2\exp(-c_3 \delta m )$ on which the following holds: there \red{exists} $J \subset [m]$ of cardinality at least $\delta m/4$ \red{such that} for every $j \in J$,
$$
|Z_j| \leq \lambda, \ \ \ |W_j| \leq \lambda, \ \ {\rm and} \ \ |Z_j-W_j| \geq \kappa \|\red{Z}-\red{W}\|_{L^1}.
$$

Now fix two sequences of numbers $(z_i)_{i=1}^m$ and $(w_i)_{i=1}^m$ and consider the independent events
$$
E_i = \{\sign(z_i+\tau_i)\neq \sign(w_i+\tau_i) \}, \ \ \ 1 \leq i \leq m.
$$
Recall that by \eqref{eqn:hypSepMinAssTauL1}, for every $i \in [m]$, if $|z_i| \leq \lambda$ and $|w_i| \leq \lambda$ then
\begin{align*}
& \bP_{\tau}(\sign(z_i+\tau_i)\neq \sign(w_i+\tau_i))
\\
& \qquad  = \bP_\tau (z_i+\tau_i>0, \ w_i+\tau_i\leq 0) + \red{\bP_{\tau}}(z_i+\tau_i \leq 0, \ w_i+\tau_i>0) \\
& \qquad  = \bP_\tau (-z_i<\tau\leq -w_i) + \bP_\tau(-w_i<\tau\leq -z_i)
\\
& \qquad \geq c_\tau |z_i-w_i|.
\end{align*}
Hence, for every realization of $(Z_i)_{i=1}^m$ and $(W_i)_{i=1}^m$ from the event ${\mathcal A}$,
$$
|\{j: \bP_{\tau}(E_j) \geq c_\tau \kappa \|Z-W\|_{L^1}\}| \geq \frac{\delta m}{4}.
$$
It follows that there are absolute constants $c_4$ and $c_5$, such that with $\tau$-probability at least $1-2\exp(-c_4 c_\tau \kappa  \delta m \|Z-W\|_{L^1})$,
$$
\sum_{i=1}^m \IND_{E_i} \geq \sum_{j \in J} \IND_{E_j} \geq \frac{|J|}{2} \cdot c_\tau \kappa \|Z-W\|_{L^1} \geq c_5 c_\tau \kappa \delta  m\|Z-W\|_{L^1}.
$$
Thus, with the \red{desired} probability with respect to $(Z_i)_{i=1}^m$, $(W_i)_{i=1}^m$ and $(\tau_i)_{i=1}^m$, one has
$$
|\{i: \sign(Z_i+\tau_i)\neq \sign(W_i+\tau_i) \}| \geq c_5 c_\tau \kappa \delta m\|Z-W\|_{L^1},
$$
as claimed.
\end{proof}

Next, let us consider the random variable $\tau$ and the random vector $X$ from Theorem~\ref{thm:tess}: $\tau\sim\cU[-\la,\la]$ and $X$ is isotropic, symmetric and satisfies an $L^1$-$L^2$ norm equivalence with constant $L$. By the Paley-Zygmund inequality (see, e.g., \cite{PeG99}) there are constants $\kappa$ and $\delta$ that depend only on $L$ for which, for every $t \in \R^n$,
$$
\bP( |\inr{X,t}| \geq \kappa \|\inr{X,t}\|_{L^1} ) \geq \delta.
$$
Therefore, $\tau$ satisfies \eqref{eqn:hypSepMinAssTauL1} with constant $\red{c_{\tau}=1/(2\lambda)}$ and the random variables $Z=\inr{X,x}$ and $W=\inr{X,w}$ satisfy Lemma~\ref{lemma:Z-W-tau} with constants $\kappa$ and $\delta$ that depend only on the equivalence constant $L$.

\vskip0.4cm

\noindent {\bf Proof of Theorem \ref{thm:individual-separation}.}
Clearly, by Lemma \ref{lemma:Z-W-tau},
$$
|\{i : \sign(\inr{X_i,x}+\tau_i) \not = \sign(\inr{X_i,y}+\tau_i)\}| \geq c(L) m \frac{\|x-y\|_2}{\lambda}
$$
with the promised probability, using the fact that
$$
\max\{ \|Z\|_{L^2}, \ \|W\|_{L^2} \} = \max \{ \|\inr{X,x}\|_{L^2}, \ \|\inr{X,y}\|_{L^2} \} \leq R.
$$
One has to show that in addition,  $|\inr{X_i,x}+\tau_i|$ and $|\inr{X_i,x}+\tau_i|$ are also reasonably large. To that end, one may apply \eqref{eqn:SBindep} twice, for $Z=\inr{X,x}$ and $Z=\inr{X,y}$, to see that for any $\eps>0$,
$$
\max\left\{\bP(|\inr{X,x}+\tau|\leq \eps), \ \bP(|\inr{X,y}+\tau|\leq \eps)\right\} \leq \frac{\eps}{\lambda}.
$$
Therefore, with probability at least $1-2\exp\left(-c \frac{\eps}{\lambda} m \right)$, there are at most $4\eps m/\lambda$ indices $i$ for which
$$
\red{\min}\left\{|\inr{X_i,x}+\tau|, \ |\inr{X_i,y}+\tau| \right\} \red{\leq} \eps;
$$
hence, setting $\eps=(c(L)/8)\|x-y\|_2$ completes the proof.
\endproof

Next, one has to use the individual high probability estimate from Theorem \ref{thm:individual-separation} to obtain a uniform estimate in $T$. The idea is to use a covering argument combined with a simple stability property:
\begin{Lemma} \label{lemma:stability}
Fix a realization of $X$ and $\tau$ and set $r^\prime>0$. Assume that $\|w-v\|_2 \geq r^\prime$ and that
$$
|\inr{X,x-v}| \leq \theta r^{\prime}/3 \ \  \ \ |\inr{X,y-w}| \leq \theta r^{\prime}/3.
$$
If $v$ and $w$ are $\theta$-well separated by $H_{X,\tau}$ then $x$ and $y$ are separated by $H_{X,\tau}$.
\end{Lemma}

\begin{proof}
Since $v$ and $w$ are $\theta$-well separated by $H_{X,\tau}$, one has
$$
\sign(\inr{X,v} + \tau) \not= \sign(\inr{X,w} + \tau), \ \  |\inr{X,v} + \tau| \geq \theta \|v-w\|_2, \ \ \ |\inr{X,w} + \tau| \geq \theta \|v-w\|_2.
$$
Therefore, if
$$
|\inr{X,x-v}| \leq \theta r^{\prime}/3 \ \ \ \ {\rm and } \ \ \ \ |\inr{X,y-w}| \leq \theta r^{\prime}/3
$$
it follows that $\sign(\inr{X,x} + \tau) \not= \sign(\inr{X,y} + \tau)$ \red{(See Figure~\ref{fig:goodHyper} for an illustration).}
\end{proof}

\begin{center}
\begin{figure}
\includegraphics[scale=0.8,trim={255 405 142 206},clip]{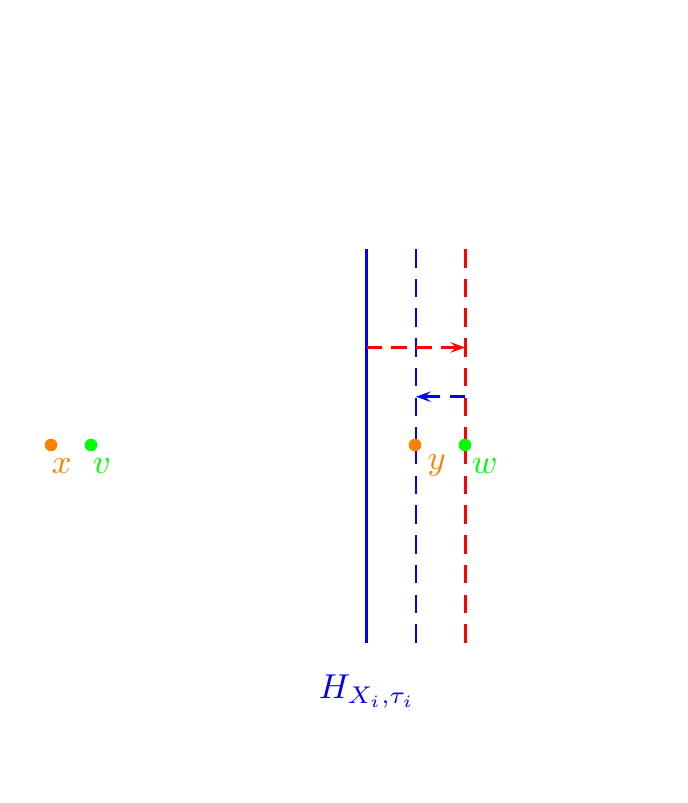}
\caption{Picture of a `good' hyperplane $H_{X_i,\tau_i}$ that well-separates $v$ and $w$. \red{On the one hand}, one needs to shift the hyperplane in parallel by a distance proportional to $\theta\|v-w\|_2$ to hit $w$ (shift marked in red). On the other hand, the parallel shift needed to hit $y$ when starting from $w$ is \red{less than half} this distance (shift marked in blue). As a consequence, a good hyperplane separates $x$ and $y$.}
\label{fig:goodHyper}
\end{figure}
\end{center}

The key component in the proof of Theorem \ref{thm:tess} is the following fact:

\begin{Theorem} \label{thm:single-to-local}
There exist constants $c_0,\ldots,c_6$ that depend only on $L$ for which the following holds. Let $\blue{\lambda \geq c_0 R}$, $r^\prime \leq \lambda/2$, and $r^{\prime \prime} \leq r^\prime/4$. Assume that
\begin{equation} \label{eq:cond-r-1}
\log {\mathcal N}(T ,r^{\prime \prime} ) \leq c_1 m \frac{r^{\prime}}{\lambda},
\end{equation}
and that
\begin{equation} \label{eq:cond-r-2}
\E \sup_{z \in (T-T) \cap r^{\prime \prime} B_2^n} |\{i: |\inr{X_i,z}| \geq c_2 r^{\prime}\}| \leq m\frac{c_3 r^\prime}{\lambda}.
\end{equation}
Then with probability at least $1-8\exp\left(-c_4 m (r^{\prime}/\lambda)^2\right)$,
for every $x,y \in T$ such that $\|x-y\|_2 \geq 2r^\prime$,
$$
|\{i : \sign(\inr{X_i,x}+\tau_i) \not=\sign(\inr{X_i,y}+\tau_i)\}| \geq c_5 \frac{ m r^\prime}{\lambda},
$$
and for every $x,y \in T$ such that $\|x-y\|_2 \leq r^{\prime \prime}/2$,
\begin{equation}
\label{eqn:locUpBdHeavy}
|\{i : \sign(\inr{X_i,x}+\tau_i) \not=\sign(\inr{X_i,y}+\tau_i)\}| \leq c_6 \frac{ mr^{\prime}}{\lambda}.
\end{equation}
\end{Theorem}

\begin{proof}
Let $V \subset T$ be an $r^{\prime \prime}$-cover of $T$. We apply \eqref{eq:single-upper} to every $Z=\inr{X,v}$, $v \in V$, and Theorem~\ref{thm:individual-separation} to every pair of points from $V$. Let $c_1\leq \min\{c,c_2\}/2$, where $c$ and $c_2$ are as in \eqref{eq:single-upper} and Theorem~\ref{thm:individual-separation}, respectively. If
\begin{equation*}
\log |V| \leq \red{c_1} \frac{m r^{\prime}}{\lambda}
\end{equation*}
then by the union bound there is an event ${\mathcal A}_1$ of probability at least
$1-6\exp\left(-\red{c_2}m r^{\prime}/\lambda\right)$ such that for every $v \in V$,
\begin{equation} \label{eq:in-proof-single-to-local-1}
|\{i : |\inr{X_i,v}+\tau_i| \geq r^{\prime}\}| \geq \left(1-\frac{2 r^\prime}{\lambda}\right)m
\end{equation}
and if $v,w \in V$ satisfy  $\|v-w\|_2 \geq r^\prime$ then
$$
|I_{v,w} (c_3)| \geq  c_4m \frac{\|v-w\|_2}{\lambda},
$$
\red{where} the constants $c_2,c_3$ and $c_4$ depend only on $L$.

Now fix $x,y \in T$ that satisfy $\|x-y\|_2 \geq 2r^\prime$ and let $v,w$ be the nearest points in $V$ to $x$ and $y$ respectively. By Lemma \ref{lemma:stability}, if $i \in I_{v,w}(c_3)$ and
$$
|\inr{X_i,x-v}|, \ |\inr{X_i,y-w}| \leq \frac{c_3}{3}r^{\prime},
$$
then $x$ and $y$ are separated by $H_{X_i,\tau_i}$.

Note that $x-v,y-w \in (T-T) \cap r^{\prime \prime} B_2^n$, let $\tilde{c}_3=\min\{c_3,1\}$ and set ${\mathcal A}_2$ to be the event
\begin{equation} \label{eq:osc-in-proof-tess}
\sup_{z \in (T-T) \cap r^{\prime \prime} B_2^n} |\{i: |\inr{X_i,z}| \geq (\tilde{c}_3/3) r^{\prime}\}| \leq \frac{c_4}{2} \cdot \frac{m r^\prime}{\lambda}.
\end{equation}
Hence, on ${\mathcal A}_1 \cap {\mathcal A}_2$, if $\|x-y\|_2 \geq r^\prime$ then
$$
|\{i : \sign(\inr{X_i,x}+\tau_i) \not=\sign(\inr{X_i,y}+\tau_i)\}| \geq \frac{c_4}{2} \cdot \frac{m r^\prime}{\lambda},
$$
which is the wanted lower bound.

At the same time, if $\|x-v\|_2 \leq r^{\prime \prime}$ then by combining \eqref{eq:in-proof-single-to-local-1} and \eqref{eq:osc-in-proof-tess}, one has the upper bound
$$
|\{i: \sign(\inr{X_i,x}+\tau_i) \not=\sign(\inr{X_i,v}+\tau_i)\}| \leq c_5\frac{m r^\prime}{\lambda}.
$$
All that is left is to estimate the probability of the event ${\mathcal A}_2$. Note that
$$
|\{i \in [m]: |\inr{X_i,w}| \geq (c_3/3) r^{\prime} \}| = \sum_{i =1}^m \IND_{\left\{|\langle X_i,w\rangle| \geq (c_3/3) r^{\prime \prime} \right\}} \red{=:} H_w,
$$
and by the bounded differences inequality \red{(see e.g.\ \cite[Theorem 6.2]{BLM13})},
$$
\bP\left( \sup_{w \in (T-T) \cap r^{\prime \prime} B_2^n} H_w \geq \E \sup_{w \in (T-T) \cap r^{\prime \prime} B_2^n} H_w + mt \right) \leq 2\exp(-cmt^2)
$$
for a suitable absolute constant $c$. The claim follows with the choice of $t=(c_4/4) \cdot  (r^\prime/\lambda)$.
\end{proof}

\par

\noindent{\bf Proof of Theorem \ref{thm:tess}.} We apply Theorem \ref{thm:single-to-local} for the choice $r^\prime =\rho/2$. Let us identify the conditions on $r^{\prime \prime}$ one has to impose to ensure that \eqref{eq:cond-r-2} is satisfied.

By the Gin\red{\'{e}}-Zinn symmetrization theorem \red{\cite{GiZ84}} and the contraction inequality for Bernoulli processes \red{\cite{LeT91}}, one has
\begin{align*}
& \E \sup_{z \in (T-T) \cap r^{\prime \prime} B_2^n} |\{i: |\inr{X_i,z}| \geq c \rho\}| \leq \E \sup_{z \in (T-T) \cap r^{\prime \prime} B_2^n} \frac{1}{c \rho} \sum_{i=1}^m |\inr{X_i,z}|
\\
\leq & \E \sup_{\red{z} \in (T-T) \cap r^{\prime \prime} B_2^n} \frac{2}{c \rho} \left|\sum_{i=1}^m \eps_i \inr{X_i,z}\right| + \frac{mr^{\prime \prime}}{c \rho}= (1)+(2).
\end{align*}
To satisfy \eqref{eq:cond-r-2} it suffices to bound both terms by $cm\rho/\lambda$. The required estimate on $(2)$ holds once
$$
r^{\prime \prime} \leq c(L) \frac{\red{\rho^2}}{\lambda},
$$
and to ensure a suitable estimate on $(1)$ it suffices that
$$
E(\red{T_{r^{\prime \prime}}}) \leq c(L) \sqrt{m} \frac{\rho^2}{\lambda}.
$$
The claim follows by setting $r=r^{\prime \prime}$.

\par

This immediately yields the lower bound in Theorem~\ref{thm:tess}. To complete the proof of the upper bound, recall that $T$ is $(r,\gamma)$-metrically convex. \blue{For given $x$, $y$ with $\|x-y\|_2\geq 2r$} let $(z_j)_{j=1}^\ell$ be as in Definition \ref{def:r-metric-convex}. Then
$$
|\{i : \sign(\inr{X_i,x}+\tau_i) \not=\sign(\inr{X_i,y}+\tau_i)\}| \leq \sum_{j=0}^\ell |\{i : \sign(\inr{X_i,z_j}+\tau_i) \not=\sign(\inr{X_i,z_{j+1}}+\tau_i)\}|,
$$
and the claim follows from the `local' upper bound \eqref{eqn:locUpBdHeavy}.
\endproof

\subsection{The subgaussian scenario}

When $X$ is an $L$-subgaussian random vector one may establish an improved version of Theorem \ref{thm:single-to-local}: first, by showing that one may take $r^{\prime \prime}$ to be of the order of $r^\prime$ up to a logarithmic factor; and second, by providing a better probability estimate on the outcome. Moreover, thanks to the subgaussian property, one may replace the empirical parameter $E(T_r)$ by its Gaussian counterpart, $\ell_*(T_r)$.

\begin{Theorem} \label{thm:single-to-local-subgaussian}
There exist constants $c_0,\ldots,c_5$ that depend only on $L$ for which the following holds. Assume that $\blue{\lambda \geq c_0 R}$,
$$
r^{\prime \prime} \leq c_1 \frac{r^\prime}{\sqrt{\log(e\lambda/r^\prime)}},
$$
that
\begin{equation} \label{eq:cond-r-1-sg}
\log {\mathcal N}(T ,r^{\prime \prime}) \leq c_2 \frac{mr^{\prime}}{\lambda},
\end{equation}
and that
\begin{equation} \label{eq:cond-r-2-sg}
\ell_*(T_{r^{\prime \prime}}) \leq c_3 \sqrt{m} \frac{(r^\prime)^{3/2}}{\sqrt{\lambda}}
\end{equation}
Then with probability at least $1-8\exp\left(-c_4 m r^{\prime}/{\lambda}\right)$,
for every $x,y \in T$ such that $\|x-y\|_2 \geq 2r^\prime$,
$$
|\{i : \sign(\inr{X_i,x}+\tau_i) \not=\sign(\inr{X_i,y}+\tau_i)\}| \geq c_5 \frac{ m r^\prime}{\lambda},
$$
and for every $x,y \in T$ such that $\|x-y\|_2 \leq r^{\prime \prime}/2$,
$$
|\{i : \sign(\inr{X_i,x}+\tau_i) \not=\sign(\inr{X_i,y}+\tau_i)\}| \leq c_6 \frac{ mr^{\prime}}{\lambda}.
$$
\end{Theorem}

The only difference between the proof of Theorem \ref{thm:single-to-local-subgaussian} and that of Theorem \ref{thm:single-to-local} is the control one has on the probability that
\begin{equation} \label{eq:osc-subgaussian-1}
\sup_{z \in (T-T) \cap r^{\prime \prime} B_2^n} |\{i: |\inr{X_i,z}| \geq c r^{\prime}\}| \leq Cm\frac{ r^\prime}{\lambda}.
\end{equation}
When $X$ merely satisfies \red{an} $L^1$-$L^2$ norm equivalence, one has to resort to the bounded differences inequality for a high probability estimate. However, when $X$ is $L$-subgaussian one has more machinery at one's disposal. Specifically, we use the following fact.
\begin{Theorem} \label{thm:monotone-subgaussian}
Let $X$ be an \red{isotropic} $L$-subgaussian random vector and let $\red{S} \subset \R^n$. If $1 \leq k \leq m$ and $u \geq 1$ then with probability at least $1-2\exp(-c_1 u^2 k \log(e\red{m}/k))$,
$$
\sup_{z \in \red{S}} \max_{|I| \leq k} \left(\sum_{i \in I} \inr{X_i,z}^2 \right)^{1/2} \leq c_2 \left(\ell_*(\red{S}) + u d_{\red{S}} \sqrt{k \log(em/k)} \right),
$$
where $c_1$ and $c_2$ depend only on $L$ and $d_{\red{S}} = \sup_{z \in \red{S}} \|z\|_2$.
\end{Theorem}
\blue{We omit the proof of Theorem \ref{thm:monotone-subgaussian}, which is standard. It is based on generic chaining (see e.g.\ \cite[Theorem 3.2]{Dir15}) combined with Talagrand's majorizing measures theorem \cite{Tal14}.}

\vskip0.4cm

\noindent {\bf Proof of Theorem \ref{thm:single-to-local-subgaussian}.}
Observe that if $(a_i)_{i=1}^m$ is a sequence of nonnegative numbers then the $k$-largest element satisfies
$$
a_k^* \leq \max_{|I| \leq k} \left(\frac{1}{k} \sum_{i \in I} a_i^2 \right)^{1/2}.
$$
With that in mind, one has to ensure that for $k = Cmr^\prime/\lambda$,
$$
\sup_{z \in (T-T) \cap r^{\prime \prime} B_2^n} \max_{|I| \leq k} \left(\frac{1}{k}\sum_{i \in I} \inr{X_i,z}^2 \right)^{1/2} \leq c r^\prime,
$$
and by Theorem \ref{thm:monotone-subgaussian}, it suffices to verify that
$$
\left(\frac{\sqrt{\lambda} \ell_*(T_{r^{\prime \prime}})}{\sqrt{m r^\prime}} + r^{\prime \prime} \sqrt{\log(e\lambda/r^\prime)} \right) \leq c_1(L)r^\prime.
$$
Clearly, the wanted estimate follows if
$$
r^{\prime \prime} \leq  c_2(L)\frac{r^\prime}{\sqrt{\log(e\lambda/r^\prime)}} \ \ \ \ {\rm and} \ \ \ \ m \geq c_3(L) \lambda \frac{(\ell_*(\red{T_{r^{\prime \prime}}}))^2}{(r^\prime)^3};
$$
and in that case, \eqref{eq:osc-subgaussian-1} holds with probability at least $1-2\exp(-c_4 m (r^\prime/\lambda) \log(e\lambda/r^\prime))$.

The rest of the proof of Theorem \ref{thm:single-to-local-subgaussian} is identical to that of Theorem \ref{thm:single-to-local} and is omitted.
\endproof

\par

Now one may complete the proof of Theorem \ref{thm:tess-subgaussian}, by setting $\rho = 2r^\prime$ and $r =r^{\prime \prime}/2$, and noting that Theorem \ref{thm:single-to-local-subgaussian} yields the lower bound and the `local' upper bound. The upper bound follows directly from the local upper bound and the metric convexity assumption (see the end of the proof of Theorem~\ref{thm:tess} for this argument).
\endproof

\section{Random tessellations - noisy measurements}
With the machinery developed for the proofs of Theorem \ref{thm:tess} and Theorem \ref{thm:tess-subgaussian} at \red{our} disposal, let us present the proofs of Theorem \ref{thm:mainIntroNC-sub} and Theorem \ref{thm:mainIntroNC-heavy}.

Recall that $T \subset R B_2^n$ and that $X$ is an isotropic and symmetric random vector, while $\tau\sim{\mathcal U}[-\la,\la]$. The rows of the measurement matrix $A$ are $(X_i)_{i=1}^m$ and the given observations are the coordinates of the vector $q_{\text{corr}}$, which is a corrupted version of
$$
\sign(Ax+\noise+\thres)= \left(\inr{X_i,x} + \nu_i + \tau_i \right)_{i=1}^m,
$$
by at most $\beta m$ `sign flips'. In the first scenario, $X$ and $\nu$ satisfy an $L^1$-$L^2$ norm equivalence with constant $L$, while in the second they are $L$-subgaussian.

\red{The goal is to show that there is a constant $C$ depending only on $\lambda$ and $L$ so that with high probability, for any $x,y \in T$ that satisfy $\|x-y\|_2 \geq \rho$,
\begin{equation} \label{eq:in-noisy-lower}
|\{i : \sign(\inr{X_i,x}+\nu_i + \tau_i) \not=\sign(\inr{X_i,y}+ \tau_i)\}| \geq C m \|x-y\|_2,
\end{equation}
and at the same time,
\begin{equation} \label{eq:in-noisy-upper-x}
|\{i : \sign(\inr{X_i,x}+\nu_i + \tau_i) \not=\sign(\inr{X_i,x}+ \tau_i)\}| < C m \rho/4.
\end{equation}
Together these conditions imply that the recovery program \eqref{eqn:bitMatchIntro}, which minimizes the Hamming distance between $q_{\text{corr}}$ and $(\sign(\inr{X_i,z}+ \tau_i))_{i=1}^m$ with respect to $z\in T$, achieves reconstruction accuracy $\rho$ as long as the fraction of the corrupted bits is at most $\beta \leq C\rho/4$. Indeed, \eqref{eq:in-noisy-upper-x} implies that any solution $x^{\#}$ of \eqref{eqn:bitMatchIntro} must satisfy
$$d_H(q_{\text{corr}},(\sign(\inr{X_i,x^{\#}}+ \tau_i))_{i=1}^m)<C m \rho/2$$
and then \eqref{eq:in-noisy-lower} shows that $\|x^{\#}-x\|_2\leq \rho$.}

\vskip0.4cm

The proofs of \eqref{eq:in-noisy-lower} in both scenarios follow from minor modifications of the results established in the previous section. Rather than repeating the arguments, let us sketch the adjustments one has to make.

\par
First, one has to consider a modified notion of being `well-separated' by a hyperplane:

\begin{Definition} \label{def:noisy-large-margin}
The hyperplane $H_{X_i,\tau_i}$ \emph{$\theta$-well-separates} $x$ and $y$ if
\begin{itemize}
\item $\sign(\inr{X_i,x}+\nu_i+\tau_i)\neq \sign(\inr{X_i,y}+ \tau_i)$,
\item $|\inr{X_i,x}+\nu_i+\tau_i|\geq \theta\|x-y\|_2$, and
\item $|\inr{X_i,y}+\tau_i|\geq \theta\|x-y\|_2$.
\end{itemize}
Denote by $J_{x,y}(\theta)\subset[m]$ the set of indices for which $H_{X_i,\tau_i}$ $\theta$-well-separates $x$ and $y$.
\end{Definition}

Next, one has to establish the analog of Theorem \ref{thm:individual-separation} and show that every pair $x,y$ is well separated by a fraction of the hyperplanes that is proportional to $\|x-y\|_2$.

\begin{Theorem} \label{thm:noisy-individual-separation}
There are constants $c_1,\ldots,c_4$ that depend only on $L$ for which the following holds. Let $x,y \in R B_2^n$ and set $\lambda \geq c_1 \max\{R,\|\nu\|_{L^2}\}$. With probability at least
$1-4\exp\left(-c_2 m \|x-y\|_2/\lambda\right)$,
$$
|J_{x,y} (c_3)| \geq  c_4 m \frac{\|x-y\|_2}{\lambda}.
$$
\end{Theorem}

Just as in Theorem \ref{thm:individual-separation}, the proof is an outcome of Lemma \ref{lemma:Z-W-tau}, only this time with the choice $Z=\inr{X,x}+\nu$ and $W=\inr{X,y}$. To that end, one has to verify that $Z-W$ satisfies a small-ball condition, and that is immediate \red{from} the following observation and the Paley-Zygmund inequality.

\begin{Lemma} \label{lemma:norm-equivalence-noisy}
Let $Z$ and $W$ be as above. Then
$$
\|Z-W\|_{L^1} \geq \frac{1}{2L} \|Z-W\|_{L^2}.
$$
\end{Lemma}

\begin{proof}
Let $\alpha, \beta \in \R$ and set $\eps$ to be a symmetric, $\{-1,1\}$-valued random variable. Clearly,
$$
\E_\eps |\eps \alpha + \beta| \geq \frac{1}{2}(|\alpha| + |\beta|).
$$
Since $X$ is a symmetric random vector, $\inr{X,x-y}$ has the same distribution as $\eps \inr{X,x-y}$, where $\eps$ that is independent of $X$ and of $\nu$. Hence,
\begin{align*}
\|Z-W\|_{L^1} & = \|\inr{X,x-y}+\nu \|_{L^1} = \E_{X,\nu} \E_\eps |\eps \inr{X,x-y} + \nu| \geq \frac{1}{2} (\E |\inr{X,x-y}| + \E |\nu|)
\\
& \geq \frac{1}{2L} (\|\inr{X,x-y}\|_{L^2} + \|\nu\|_{L^2}) \geq \frac{1}{2L} \|\inr{X,x-y}+\nu\|_{L^2}.
\end{align*}
\end{proof}

The other components needed for ensuring separation in the sense of Definition \ref{def:noisy-large-margin} follow from an identical argument used in the proof of Theorem \ref{thm:individual-separation}, by conditioning on $X_i$ and $\nu_i$ rather than just on $X_i$.

\par

Theorem \ref{thm:noisy-individual-separation} allows one to control the set of `centers' $V$ of a cover of $T$, and all that remains now is to show that if $x^\prime$ is close to a center $x$ and $y^\prime$ is close $y$ then there will be few indices for which
$$
\sign(\inr{X_i,x}+\nu_i + \tau_i) \not=\sign(\inr{X_i,x^\prime}+\nu_i + \tau_i),
$$
or
$$
\sign(\inr{X_i,y}+ \tau_i) \not=\sign(\inr{X_i,y^\prime}+ \tau_i).
$$
In both cases, and using the notation of the previous section, one may follow the argument in the proof of Theorem~\ref{thm:single-to-local}. Thus, it suffices to show that
\begin{equation} \label{eq:osc-in-proof-noisy-tess}
\sup_{z \in (T-T) \cap r^{\prime \prime} B_2^n} |\{i: |\inr{X_i,z}| \geq c_1 r^{\prime}\}| \leq c_2 m\frac{r^\prime}{\lambda},
\end{equation}
and that concludes the proof of the bound \eqref{eq:in-noisy-lower} with $C \red{\sim_L} \frac{1}{\lambda}$.
\endproof

Let us turn to \eqref{eq:in-noisy-upper-x}. In the heavy-tailed case, one may invoke the proof of Theorem \ref{thm:single-to-local} to show that with probability at least $1-8\exp(-cm(\rho/\lambda)^2)$, for every $x \in T$,
$$
\left|\left\{ i : |\inr{X_i,x} + \tau_i | \geq \frac{\rho}{2} \right\} \right| \geq m \left(1-\frac{\rho}{\lambda}\right).
$$
If $|\E \nu| \leq \rho/16$ and $\sigma^2 \leq (1/64) \rho^3/\lambda$, it follows that
$$
\bP(|\nu| \geq \rho/4) \leq \bP( |\nu - \E \nu| \geq \rho/8) \leq \frac{\rho}{\lambda}.
$$
Hence, with probability at least $1-2\exp(-cm\rho/\lambda)$,
$$
|\{i: |\nu_i| \geq \rho/4\}| \leq \frac{2\rho m}{\lambda}.
$$
On the intersection of the two events, for every $x \in T$ one has
$$
|\{i : \sign(\inr{X_i,x} + \nu_i + \tau_i) \not = \sign(\inr{X_i,x} + \tau_i)\}| \leq \frac{3 \rho m}{\lambda},
$$
as required.

\par

The proof in the subgaussian case is \red{analogous and therefore omitted}.
\endproof

\section{Robust recovery via a convex program} \label{sec:robust-convex}
This section is devoted to the proof of Theorem \ref{thm:convex-iid-rows}. Set
\begin{equation} \label{eq:reg-functional}
\phi(z) = \frac{1}{m}\inr{q_{\rm corr},Az} - \frac{1}{2\lambda} \|z\|_2^2,
\end{equation}
\red{then the convex optimization procedure \eqref{eqn:convTintro} is exactly}
$$
\max_{z \in \co(T)} \phi(z).
$$
Recall that $U=\co(T)$ and that $U_\rho = (U-U) \cap \rho B_2^n$; $X_1,...,X_m$ are the rows of the matrix $A$ which are distributed according to an isotropic, symmetric, $L$-subgaussian random vector $X$; and $\tau \sim {\mathcal U}[-\la,\la]$. Here we assume for the sake of simplicity that $\nu$ has mean zero and variance $\sigma^2$, though the modifications needed to handle the case in which $\nu$ has a nontrivial adversarial component are straightforward. Finally, as before $q_{\text{corr}}\in \{-1,1\}^m$ satisfies
\begin{equation} \label{eqn:fractionCorr}
d_H(q_{\text{corr}},\sign(Ax+\noise+\thres)) \leq \beta m.
\end{equation}

\red{As in most regularized procedures, the idea is to study the `excess functional' $\phi(z)-\phi(x)$: for a reconstruction error $\rho$, we determine a sufficient condition on the number of measurements $m$ which guarantees that $\phi(z)-\phi(x) < 0$ whenever $z \in U$ and $\|x-z\|_2 \geq \rho$. Clearly, that implies that the solution $x^{\#}$ to \eqref{eq:reg-functional} satisfies $\|x^{\#}-x\|_2 \leq \rho$. As before, we wish to obtain a uniform estimate, i.e., the high probability event under which the above holds should not depend on the identity of $x \in T$.}

\par

The first step towards a uniform estimate is a decomposition of the excess functional. Note that
\begin{align*}
\phi(z)-\phi(x) & = \frac{1}{m} \left(\inr{q_{\operatorname{corr}},Az}-\inr{q_{\operatorname{corr}},Ax}\right) -\frac{1}{2\lambda}\|z\|_2^2 + \frac{1}{2\lambda}\|x\|_2^2
\\
& = \frac{1}{m} \inr{q_{\operatorname{corr}}-\sign{(Ax+\noise+\thres)},A(z-x)}
\\
& + \frac{1}{m} \left(\inr{\sign{(Ax+\noise+\thres)},A(z-x)} - \E \inr{\sign{(Ax+\noise+\thres)},A(z-x)} \right)
\\
& + \frac{1}{m} \E \inr{\sign{(Ax+\noise+\thres)},A(z-x)} -\frac{1}{2\lambda}\|z\|_2^2 + \frac{1}{2\lambda}\|x\|_2^2
\\
& \red{=:} (1)+(2)+(3),
\end{align*}

We use this decomposition to find constants $C$ and $\rho>0$ and a high probability event on which, for every $x \in T$ and $z \in U$,
\begin{equation} \label{eq:path}
|(1)|  \leq C\|x-z\|_2^2; \ \ |(2)| \leq C\|x-z\|_2^2 \ \ {\rm and} \ \ (3) \leq -4C\|x-z\|_2^2,
\end{equation}
provided that $\|x-z\|_2 \geq \rho$.

\vskip0.4cm

\subsection*{Estimating $(3)$}

The starting point is a straightforward observation: for $\tau \sim \mathcal{U}[-\la,\la]$ and any $z \in \R$,
\begin{equation} \label{eq:uniform+const-in-conv}
\E\sign(z+\tau) = \frac{z}{\lambda}\IND_{\{|z| \leq \lambda\}}+\IND_{\{z>\lambda\}}-\IND_{\{z<-\lambda\}}.
\end{equation}

\begin{Lemma} \label{lemma:est-on-(3)-iid}
There exist absolute constants $C$ and $c$ for which the following holds. Let $Z$ and $W$ be random variables and let $\tau \sim \mathcal{U}[-\la,\la]$ be independent of $Z$ and $W$.  Then
$$
\left|\E W\sign(Z+\tau) -\frac{1}{\la} \E W Z \right| \leq C\|W\|_{L_2} \max\left\{1,\frac{\|Z\|_{\psi_2}}{\la}\right\} \exp(-c \la^2/\|Z\|_{\psi_2}^2).
$$
\end{Lemma}

\proof
By \eqref{eq:uniform+const-in-conv},
\begin{align*}
\E_\tau W\sign(Z+\tau) = & W \left(\frac{Z}{\la} \IND_{\{|Z| \leq \la\}} + \IND_{\{Z > \la\}} - \IND_{\{Z < - \la\}} \right)
\\
= & W \left(\frac{Z}{\la} - \frac{Z}{\la}\IND_{\{|Z|> \la\}} + \IND_{\{Z > \la\}} - \IND_{\{Z < - \la\}} \right)
\\
= & \frac{WZ}{\la} + (*).
\end{align*}
Hence, $\E \sign(Z+\tau)W = \frac{1}{\la} \E W Z + \E(*)$, and all that is left to show is
$$
\E|(*)| \leq C \|W\|_{L_2} \max\left\{1,\frac{\|Z\|_{\psi_2}}{\la}\right\}  \exp(-c\la^2/\|Z\|_{\psi_2}^2)
$$
for absolute constants $C$ and $c$.

Note that $\E |WZ\IND_{\{|Z| >\la\}} | \leq \|W\|_{L_2} (\E Z^2\IND_{\{|Z| >\la\}})^{1/2}$. By tail integration,
$$
\E Z^2 \IND_{\{|Z|> \la\}} \leq \la^2 \bP(|Z| > \la) + 2\int_\la^\infty t\bP(|Z|>t)dt \leq (\la^2+\|Z\|_{\psi_2}^2) \exp(-c_1\la^2/\|Z\|_{\psi_2}^2),
$$
where $c_1$ is a suitable absolute constant. The estimate on the other two terms follows because $\IND_{\{|Z|>\la\}} \leq (|Z|/\la)\IND_{\{|Z|>\la\}}$.
\endproof

\begin{Corollary} \label{cor:(3)-in-convex-a}
There exist absolute constants $c$ and $C$ for which the following holds. For every $x,z \in \R^n$,
\begin{align*}
& \frac{1}{m} \E \inr{\sign(Ax+\noise+\thres),A(z-x)}
\\
\leq & \frac{1}{\la} \inr{x,z-x} + C\|z-x\|_2 \max\left\{1,\frac{L (\sigma+\|x\|_2)}{\la}\right\} \exp\left(-c\frac{ \la^2}{L^2(\|x\|_2^2+\sigma^2)}\right),
\end{align*}
where, as always, $L$ is the subgaussian constant of $X$ and of $\nu$.

In particular,
\begin{align*}
& \frac{1}{m} \E \inr{\sign{(Ax+\noise+\thres)},A(z-x)} -\frac{1}{2\la}\|z\|_2^2 + \frac{1}{2\la}\|x\|_2^2
\\
\leq & -\frac{1}{2\la} \|z-x\|_2^2 + C\|z-x\|_2 \max\left\{1,\frac{L (\sigma+R)}{\la}\right\} \exp(-c \la^2/L^2(R^2+\sigma^2))
\end{align*}
\end{Corollary}

\proof
The first statement follows from Lemma \ref{lemma:est-on-(3)-iid} for the choices of $Z_i = \inr{X_i,x} + \nu_i$ and $W_i=\inr{X_i,z-x}$: recalling that the $\nu_i$ are centred, have variance $\sigma^2$ and are independent of $X_i$,
$$
\|Z_i\|_{\psi_2}^2 \leq c (\|\inr{X_i,x}\|_{\psi_2}^2 + \|\nu_i\|_{\psi_2}^2) \leq cL^2 (\|x\|_2^2+\sigma^2) \leq cL^2 (R^2+\sigma^2);
$$
and, because $X$ is isotropic,
$$
\|\langle X_i,z-x\rangle\|_{L_2} = \|z-x\|_2 \ \ \ {\rm and} \ \ \ \frac{1}{\la} \E (\langle X_i,x\rangle+\nu_i)\langle X_i,z-x\rangle =  \frac{1}{\la} \langle x,z-x\rangle.
$$
The `in particular' part is evident because
$$
\frac{1}{\la}\inr{x,z-x}  - \frac{1}{2\la}\|z\|_2^2 + \frac{1}{2\la}\|x\|_2^2 = -\frac{1}{2\la} \|z-x\|_2^2.
$$
\endproof

Corollary \ref{cor:(3)-in-convex-a} leads to the wanted estimate on $(3)$:

\begin{Corollary} \label{cor:(3)-in-convex}
There are constants $c_1$ and $c_2$ that depend only on $L$ such that if
\begin{equation} \label{eq:condition-on-la-(1)}
\lambda \geq c_1 (\sigma+R) \sqrt{\log(c_2/\rho)}
\end{equation}
and $\|x-z\|_2 \geq \rho$ then
\begin{equation} \label{eq:(3)-outcome}
\frac{1}{m} \E \inr{\sign{(Ax+\noise+\thres)},A(z-x)} -\frac{1}{2\lambda}\|z\|_2^2 + \frac{1}{2\lambda}\|x\|_2^2 \leq -\frac{1}{4\lambda} \|z-x\|_2^2.
\end{equation}
\end{Corollary}

\subsection*{Estimating $(1)$}

Next, let us estimate $|(1)|$ from above, by studying
$$
\sup_{x \in T} \sup_{\{z \in U : \|z-x\|_2 \geq \rho \}} \left|\frac{1}{m} \inr{q_{\operatorname{corr}}-\sign{(Ax+\noise+\thres)},A(z-x)/\|z-x\|_2^2}\right|=(*).
$$
Observe that for every $x,z$, and $\eta_i=(q_{\operatorname{corr}})_i-\sign(\langle X_i,x\rangle+\nu_i+\tau_i)$, one has that
$$
\left|\frac{1}{m} \sum_{i=1}^m \eta_i \inr{X_i,z-x} \right| \leq \frac{2}{m} \sum_{ \{i : \eta_i \not = 0\}} |\inr{X_i,z-x}|
$$
and $(\eta_i)_{i=1}^m$ has at most $\beta m$ non-zero coordinates. Hence, for any $X_1,...,X_m$,
$$
(*) \leq \max_{|I| \leq \beta m} \sup_{\{x,z \in U, \|z-x\|_2 \geq \rho\}} \frac{1}{m} \sum_{i \in I} |\inr{X_i,(z-x)/\|z-x\|_2^2}|.
$$
Therefore, taking into account the estimate \eqref{eq:(3)-outcome} on $(3)$, it suffices to show that for every $x,z \in U$ such that $\|x-z\|_2 \geq \rho$,
$$
\max_{|I| \leq \beta m} \frac{1}{m} \sum_{i \in I} |\inr{X_i,(z-x)/\|z-x\|_2^2}| \leq \frac{1}{16\la}.
$$
\red{To that end, observe the following: if $f:\R^n\to \R_+$ is positive homogeneous and $W\subset \R^n$ is star-shaped around $0$, i.e., $\theta w\in W$ for all $w\in W$ and $0<\theta<1$, then
$$\sup_{\{w\in W \ : \ \|w\|_2\geq \rho\}} f(w/\|w\|_2^2) = \sup_{\{w\in W \ : \ \|w\|_2=\rho\}} f(w)/\rho^2.$$
We will refer to this argument, which reflects the general fact that star-shaped sets become `relatively richer' close to their centres, as a `star-shape argument'.}

\begin{Theorem} \label{thm:(1)-in-iid-convex}
There exist constants $c_1,c_2,$ and $c_3$ depending only on $L$ for which the following holds. Assume that
\begin{equation} \label{eq:condtion-on-rho-(1)}
\rho \geq  c_1 \lambda \beta  \sqrt{\log(e/\beta)} \ \ {\rm and} \ \ \ell_*(U_\rho) \leq c_2 \sqrt{\frac{m}{\beta}} \cdot \frac{\rho^2}{ \lambda}.
\end{equation}
Then with probability at least $1-2\exp(-c_3 \beta m \log(e/\beta))$, for every $x,z \in U$ such that $\|z-x\|_2 \geq \rho$ one has
$$
\max_{|I| \leq \beta m} \frac{1}{m} \sum_{i \in I} |\inr{X_i,z-x}| \leq \frac{1}{16\la} \|z-x\|_2^2.
$$
\end{Theorem}

\begin{proof}
\red{Since $U-U$ is star-shaped around $0$, the above star-shape argument implies that for any $I \subset [m]$}
$$
\sup_{\{x,z \in U, \|z-x\|_2 \geq \rho\}} \frac{1}{m } \sum_{i \in I} \Big|\Big\langle X_i,\frac{z-x}{\|z-x\|_2^2}\Big\rangle\Big| = \sup_{w \in (U-U)\cap \rho S^{n-1}} \frac{1}{m \rho^2} \sum_{i \in I} \left|\inr{X_i,w}\right|.
$$
Apply Theorem \ref{thm:monotone-subgaussian} for $k = \beta m$ and note that for $w \in \R^k$, $\|w\|_1 \leq \sqrt{k} \|w\|_2$. This shows that, with probability at least $1-2\exp(-c(L) \beta m \log(e/\beta))$,
$$
\sup_{w \in (U-U)\cap \rho S^{n-1}} \frac{1}{m \rho^2} \sum_{i \in I} \left|\inr{X_i,w}\right|  \leq  \frac{C(L)}{\red{m}\rho^2} \left(\sqrt{\beta m} \ell_*(U_\rho) +\rho \beta m \sqrt{\log(e/\beta)}\right) \leq \frac{1}{16 \la},
$$
provided that \eqref{eq:condtion-on-rho-(1)} holds.
\end{proof}

\subsection*{Estimating $(2)$}

Finally, let us derive an upper estimate on
$$
\left|\frac{1}{m} \left(\inr{\sign{(Ax+\noise+\thres)},A(z-x)} - \E \inr{\sign{(Ax+\noise+\thres)},A(z-x)} \right)\right|
$$
that holds uniformly for all $x \in T$ and $z \in U$ such that $\|x-z\|_2 \geq \rho$.
\begin{Theorem} \label{thm:(2)-in-iid-convex}
There exist constants $c_1,\ldots,c_4$ that depend only on $L$ for which the following holds. Let $\lambda \geq c_1$ and assume that
$$
\log{\mathcal N}\left(T, \rho \log^{-1}(e \lambda/\rho)\right) \leq c_2 m \cdot \red{\frac{\rho^2}{\lambda^2}} \qquad \text{and} \qquad \ell_*(U_\rho) \leq c_3 \frac{\rho^2}{\lambda}\sqrt{m}.
$$
Then with probability at least \red{$1-5\exp(c_4m\rho^2/\lambda^2)$}, we have that for every $x\in T$ and $z \in U$ such that $\|x-z\|_2 \geq \rho$,
$$
|(2)| \leq \frac{\rho^2}{16 \lambda}.
$$
\end{Theorem}

The proof of Theorem \ref{thm:(2)-in-iid-convex} is based on a covering argument. Let $V \subset T$ be an $r$-cover for a well-chosen $r$. A crucial part of the proof is to show that for every $\red{v} \in V$ the (random) set of sign patterns
$$
\mathbb{S}_{\red{v}} = \left\{ \sign{(A\red{x}+\noise\red{+}\thres)} : \red{x} \in U, \|\red{x-v}\|_2 \leq r \right\}
$$
is relatively simple: it consists of small perturbations of $\sign{(A\red{v}+\noise+\thres)}$.

Since this observation is essentially the second part of Theorem \ref{thm:single-to-local-subgaussian}, we formulate it in the way that it will be applied and omit its proof---which is almost identical to the proof of Theorem \ref{thm:single-to-local-subgaussian}.

\begin{Lemma} \label{lemma-strucutre-S_x}
There are constants $c_0, c_1$ and $c_2$ that depend only on $L$ for which the following holds. Let $0< r^\prime \leq \lambda/2$ and set $r^{\prime \prime} \leq c_0r^\prime /\sqrt{\log(e\lambda/r^\prime)}$. Assume that
\begin{equation} \label{eq:condition-on-eps-(2)}
\log \mathcal{N}(T,r^{\prime \prime} ) \leq c_1 \frac{r^\prime}{\lambda} m,
\end{equation}
and that
\begin{equation} \label{eq:condition-on-eps-r-(2)}
\ell_*(U_{r^{\prime \prime}}) \leq c_1 \frac{(r^\prime)^{3/2}}{\lambda^{1/2}} \sqrt{m}.
\end{equation}
If $V$ is a minimal $r^{\prime \prime}$ cover of $T$, then with probability at least $1-2\exp(-c_2 m r^\prime/\lambda)$, for every $\red{v} \in V$
$$
\mathbb{S}_{\red{v}} \subset \sign{(A\red{v}+\red{\noise}+\red{\thres})} + \blue{2}\mathcal{Z},
$$
where $\mathcal{Z}$ is the set of all $\{-1,0,1\}$-valued vectors in $\R^m$ that have at most $3(r^\prime/\lambda)m$ non-zero coordinates.
\end{Lemma}

\par

\noindent {\bf Proof of Theorem \ref{thm:(2)-in-iid-convex}.} By a standard symmetrization argument it suffices to estimate the probability with which
$$
\sup_{x \in T} \sup_{\{z \in U: \|z-x\|_2 \geq \rho\}} \left|\frac{1}{m} \sum_{i=1}^m \eps_i \sign{(\inr{X_i,x}+\nu_i+\tau_i)}\left\langle X_i,\frac{x-z}{\|x-z\|_2^2}\right\rangle\right| \leq \frac{1}{32 \lambda},
$$
where $(\eps_i)_{i=1}^m$ are independent, symmetric, $\{-1,1\}$-valued random variables that are also independent of $(X_i,\nu_i,\tau_i)_{i=1}^m$.

\red{Let $r^\prime$, $r^{\prime \prime}$, and $V$ be as in Lemma \ref{lemma-strucutre-S_x} and observe that}
\begin{align*}
& \sup_{x \in T} \sup_{\{z \in U : \|z-x\|_2 \geq \rho\}} \left|\frac{1}{m} \sum_{i=1}^m \eps_i \sign{(\inr{X_i,x}+\nu_i+\tau_i)}\left\langle X_i,\frac{x-z}{\|x-z\|_2^2}\right\rangle\right|
\\
= & \max_{\red{v} \in V} \sup_{\{\red{x} \in T : \|\red{x-v}\|_2 \leq r^{\prime \prime}\}} \sup_{\{z \in U : \|z-\red{x}\|_2 \geq \rho\}} \left|\frac{1}{m} \sum_{i=1}^m \eps_i \sign{(\inr{X_i,\red{x}}+\nu_i+\tau_i)}\left\langle X_i,\frac{\red{x}-z}{\|\red{x}-z\|_2^2}\right\rangle\right|.
\end{align*}
Let ${\mathcal A}$ be the event from Lemma \ref{lemma-strucutre-S_x}. Recall that $\red{\bP}({\mathcal A}) \geq 1-2\exp(-cmr^\prime/\lambda)$ and that on ${\mathcal A}$, for any $\red{v} \in V$ and $\red{x} \in T$ such that $\|\red{x-v}\|_2 \leq r^{\prime \prime}$, $\sign{(\inr{X_i,\red{x}}+\nu_i+\tau_i)}$ differs from $\sign{(\inr{X_i,\red{v}}+\nu_i+\tau_i)}$ on at most $(3r^\prime/\lambda)m$ indices. Therefore, for every $\red{v} \in V$,
\begin{align*}
& \sup_{\{\red{x} \in T : \|\red{x}-\red{v}\|_2 \leq r^{\prime \prime}\}} \sup_{\{z \in U : \|z-\red{x}\|_2 \geq \rho\}} \left|\frac{1}{m} \sum_{i=1}^m \eps_i \sign{(\inr{X_i,\red{x}}+\nu_i+\tau_i)}\left\langle X_i,\frac{\red{x}-z}{\|\red{x}-z\|_2^2}\right\rangle\right|
\\
\leq & \sup_{\{\red{x} \in T : \|\red{x}-\red{v}\|_2 \leq r^{\prime \prime}\}} \sup_{\{z \in U : \|z-\red{x}\|_2 \geq \rho\}} \left|\frac{1}{m} \sum_{i=1}^m \eps_i \sign{(\inr{X_i,\red{v}}+\nu_i+\tau_i)}\left\langle X_i,\frac{\red{x}-z}{\|\red{x}-z\|_2^2}\right\rangle \right| +
\\
+ & 2\sup_{\{\red{x} \in T : \|\red{x}-\red{v}\|_2 \leq r^{\prime \prime}\}} \sup_{\{z \in U : \|z-\red{x}\|_2 \geq \rho\}} \max_{|I| \leq (3r^\prime/\lambda)m} \frac{1}{m} \sum_{i \in I} \left|\left\langle X_i,\frac{\red{x}-z}{\|\red{x}-z\|_2^2}\right\rangle\right|
\\
=: & (a)_{\red{v}}+(b)_{\red{v}}.
\end{align*}
Observe that both $(a)_{\red{v}}$ and $(b)_{\red{v}}$ are homogeneous in $\red{x}-z$; therefore, by a star-shape argument
\begin{equation}
\label{eqn:estavStarShape}
(a)_{\red{v}} \red{\leq} \frac{1}{\rho^2}  \sup_{\{z,\red{x} \in U : \|z-\red{x}\|_2 = \rho\}} \left|\frac{1}{m} \sum_{i=1}^m \eps_i \sign{(\inr{X_i,{\red{v}}}+\nu_i+\tau_i)}\inr{X_i,\red{x}-z}\right|
\end{equation}
and
\begin{equation}
\label{eqn:estbvStarShape}
(b)_{\red{v}} \red{\leq} \frac{1}{\rho^2}  \sup_{\{z,\red{x} \in U : \|z-\red{x}\|_2 = \rho\}} \max_{|I| \leq (3r^\prime/\lambda)m} \frac{1}{m} \sum_{i \in I} \left|\inr{X_i,\red{x}-z}\right|.
\end{equation}

Let us begin by estimating the right hand side of \eqref{eqn:estavStarShape}. For every fixed $\red{v\in V}$ and conditioned on $(X_i,\nu_i,\tau_i)_{i=1}^m$, this is the supremum of a Bernoulli process indexed by a set of the form
$$
\left\{ \sum_{i=1}^m a_i w_i : w \in W\right\},
$$
where $(a_i)_{i=1}^m$ is a fixed vector of signs and
$$
W \subset \left\{ \left(\inr{X_i,\red{u}}\right)_{i=1}^m : \red{u} \in U_\rho \right\}.
$$
By the contraction inequality for Bernoulli processes \red{\cite{LeT91}} applied conditionally on $(X_i,\nu_i,\tau_i)_{i=1}^m$, it follows that
$$
\bP \left( \sup_{u \in U_\rho} \left|\sum_{i=1}^m \eps_i a_i \inr{X_i,u} \right| \geq t \right) \leq 2\bP \left( \sup_{u \in U_\rho} \left|\sum_{i=1}^m \eps_i \inr{X_i,u} \right| \geq t \right).
$$
\blue{Since $\frac{1}{\sqrt{m}} \sum_{i=1}^m \eps_i X_i$ is $cL$-subgaussian, it follows by generic chaining (see e.g.\ \cite[Theorem 3.2]{Dir15}) combined with Talagrand's majorizing measures theorem \cite{Tal14} that}
$$
\left\|\sup_{u \in U_\rho} \frac{1}{\sqrt{m}} \sum_{i=1}^m \eps_i \inr{X_i,u} \right\|_{L_p} \leq c(L) \left(\ell_*(U_\rho) + \sqrt{p} \sup_{u \in U_\rho} \|u\|_2\right).
$$
\red{Hence, for any $t\geq 1$, with probability at least $1-e^{-t}$,
$$\frac{1}{\rho^2} \sup_{u \in U_\rho} \left|\frac{1}{m} \sum_{i=1}^m \eps_i \inr{X_i,u} \right| \leq \frac{c_0}{\rho^2}\left(\frac{\ell_*(U_\rho)}{\sqrt{m}} + \rho\sqrt{\frac{t}{m}}\right).$$
Taking $t=c_1 m\rho^2/\lambda^2$, it follows that if
\begin{equation} \label{eq:condition-on-rho-(2)-a}
\ell_*(U_\rho) \lesssim_L \frac{\rho^2 }{\lambda}\sqrt{m},
\end{equation}
then for any fixed $v\in V$,
\begin{equation} \label{eq:(a)-in(2)}
(a)_v \leq \frac{1}{64\la} \ \ \ {\rm with \ probability \ } 1-\exp(-c_1 m\rho^2/\lambda^2).
\end{equation}
By the union bound, \eqref{eq:(a)-in(2)} holds uniformly for all $v\in V$ as long as
\begin{equation} \label{eq:condition-on-rho-(2)-b}
\log |V| = \log {\mathcal N}(T,r^{\prime \prime}) \leq \frac{c_1}{2} \frac{m\rho^2}{\lambda^2}.
\end{equation}
}
Next, observe that by \eqref{eqn:estbvStarShape}
$$
\sup_{\red{v} \in V} (b)_{\red{v}} \leq \frac{1}{\rho^2}\sup_{u \in U_\rho} \max_{|I| \leq (3r^\prime/\lambda)m} \frac{1}{m} \sum_{i \in I} \left|\inr{X_i,u}\right|.
$$
By Theorem \ref{thm:monotone-subgaussian} for $k=3r^\prime m/\lambda$, with probability at least $1-2\exp(-c_2(r^\prime/\lambda)\log(e\lambda/r^\prime)m)$,
$$
\frac{1}{\rho^2} \sup_{u \in U_\rho} \max_{|I| \leq (3r^\prime/\lambda)m} \frac{1}{m} \sum_{i \in I} \left|\inr{X_i,v}\right| \lesssim_L \left( \sqrt{\frac{r^\prime}{\lambda m}} \frac{\ell_*(U_\rho)}{\rho^2} +  \frac{r^\prime}{\lambda \rho}\sqrt{ \log(e\lambda/r^\prime)} \right) \leq \frac{1}{\red{64} \lambda},
$$
where the last inequality holds as long as
\begin{equation} \label{eq:condition-on-eps-rho-(2)}
\ell_*(U_\rho) \lesssim_L \frac{\rho^2}{\sqrt{r^\prime \lambda}} \sqrt{m}, \ \ \ \ {\rm and} \ \ \ \ r^\prime \sqrt{\log(e\lambda/r^\prime)} \lesssim_L \rho.
\end{equation}
All that is left is to select $r^\prime$ and $r^{\prime \prime}$, taking into account the conditions accumulated along the way, specifically, \eqref{eq:condition-on-eps-(2)}; \eqref{eq:condition-on-eps-r-(2)}; \eqref{eq:condition-on-rho-(2)-a}; \eqref{eq:condition-on-rho-(2)-b} and \eqref{eq:condition-on-eps-rho-(2)}.

Our starting point is \eqref{eq:condition-on-rho-(2)-a}, which is a condition on $\rho$ and does not involve $r^\prime$ or $r^{\prime \prime}$. Next, \red{the second condition in \eqref{eq:condition-on-eps-rho-(2)} is satisfied if we set}
\begin{equation} \label{eq:choice-of-eps-(2)}
r^\prime \simeq_L \frac{\rho}{\sqrt{\log(e\lambda/\rho)}},
\end{equation}
and we may assume without loss of generality that $r^\prime \leq \rho$ by the choice of $\lambda \geq c(L)$.

\red{With this choice for $r^{\prime}$, the first condition in \eqref{eq:condition-on-eps-rho-(2)} holds if}
$$
\ell_*(U_\rho) \lesssim_L \frac{\rho^{3/2}}{\sqrt{\lambda}} \log^{1/4}(e\lambda/\rho) \cdot \sqrt{m},
$$
which is automatically satisfied if \eqref{eq:condition-on-rho-(2)-a} holds.

With the choice of $r^\prime$ set in place, we set $r^{\prime\prime}$ according to the condition in Lemma~\ref{lemma-strucutre-S_x}, i.e.,
\begin{equation} \label{eq:choice-of-r-(2)}
r^{\prime \prime} \simeq_L \frac{r^{\prime}}{\sqrt{\log(e\lambda/r^{\prime})}}\simeq_L \frac{\rho}{ \log(e\lambda/\rho)}.
\end{equation}
Moreover, since $r^{\prime \prime} \leq \rho$, \eqref{eq:condition-on-eps-r-(2)} holds if
$$
\ell_*(U_\rho) \lesssim_L \frac{\rho^{3/2}}{\lambda^{1/2} \red{\log^{3/4}}(e\lambda/\rho)}\red{\sqrt{m}},
$$
which is satisfied thanks to \eqref{eq:condition-on-rho-(2)-a} and the choice $\lambda\geq c(L)$.

\red{Finally, to satisfy \eqref{eq:condition-on-eps-(2)} we need}
$$
\log |V| = \log {\mathcal N}(T,r^{\prime \prime} ) \lesssim_L \frac{r^\prime}{\lambda} m \simeq_L \frac{m\rho}{\lambda \sqrt{\log (e\lambda/\rho)}},
$$
\red{which is true by \eqref{eq:condition-on-rho-(2)-b}.}
\endproof
\vskip0.4cm

The proof of Theorem \ref{thm:convex-iid-rows} is concluded by combining the estimates on $(1)$,  $(2)$ and $(3)$.
\endproof


\begin{thebibliography}{10}

\bibitem{ALP14}
A.~Ai, A.~Lapanowski, Y.~Plan, and R.~Vershynin.
\newblock One-bit compressed sensing with non-{G}aussian measurements.
\newblock {\em Linear Algebra Appl.}, 441:222--239, 2014.

\bibitem{AGM15}
S.~Artstein-Avidan, A.~Giannopoulos, and V.~D. Milman.
\newblock {\em Asymptotic geometric analysis, Part I}, volume 202.
\newblock American Mathematical Soc., 2015.

\bibitem{BFN17}
R.~G. Baraniuk, S.~Foucart, D.~Needell, Y.~Plan, and M.~Wootters.
\newblock Exponential decay of reconstruction error from binary measurements of
  sparse signals.
\newblock {\em IEEE Transactions on Information Theory}, 63(6):3368--3385,
  2017.

\bibitem{BLM13}
S.~Boucheron, G.~Lugosi, and P.~Massart.
\newblock {\em Concentration Inequalities: A Nonasymptotic Theory of
  Independence}.
\newblock Oxford University Press, 2013.

\bibitem{BoB08}
P.~T. Boufounos and R.~G. Baraniuk.
\newblock 1-bit compressive sensing.
\newblock In {\em 2008 42nd Annual Conference on Information Sciences and
  Systems}, pages 16--21. IEEE, 2008.

\bibitem{PeG99}
V.~de~la Pe{\~n}a and E.~Gin{\'e}.
\newblock {\em Decoupling}.
\newblock Springer-Verlag, New York, 1999.

\bibitem{Dir15}
S.~Dirksen.
\newblock Tail bounds via generic chaining.
\newblock {\em Electron. J. Probab.}, 20:no. 53, 1--29, 2015.

\bibitem{DJR17}
S.~Dirksen, H.~C. Jung, and H.~Rauhut.
\newblock One-bit compressed sensing with {G}aussian circulant matrices.
\newblock {\em ArXiv:1710.03287}, 2017.

\bibitem{DLR16}
S.~Dirksen, G.~Lecu{\'e}, and H.~Rauhut.
\newblock On the gap between restricted isometry properties and sparse recovery
  conditions.
\newblock {\em IEEE Trans. Inform. Theory}, to appear.
\newblock ArXiv:1504.05073.

\bibitem{DiM18b}
S.~Dirksen and S.~Mendelson.
\newblock Robust one-bit compressed sensing with partial circulant matrices.
\newblock {\em in preparation}.

\bibitem{GiZ84}
E.~Gin{\'e} and J.~Zinn.
\newblock Some limit theorems for empirical processes.
\newblock {\em The Annals of Probability}, pages 929--989, 1984.

\bibitem{GrN98}
R.~M. Gray and D.~L. Neuhoff.
\newblock Quantization.
\newblock {\em IEEE transactions on information theory}, 44(6):2325--2383,
  1998.

\bibitem{JLB13}
L.~Jacques, J.~N. Laska, P.~T. Boufounos, and R.~G. Baraniuk.
\newblock Robust 1-bit compressive sensing via binary stable embeddings of
  sparse vectors.
\newblock {\em IEEE Trans. Inform. Theory}, 59(4):2082--2102, 2013.

\bibitem{KSW16}
K.~Knudson, R.~Saab, and R.~Ward.
\newblock One-bit compressive sensing with norm estimation.
\newblock {\em IEEE Trans. Inform. Theory}, 62(5):2748--2758, 2016.

\bibitem{LeM14}
G.~Lecu{\'e} and S.~Mendelson.
\newblock Sparse recovery under weak moment assumptions.
\newblock {\em Journal of the European Mathematical Society}, 19(3):881--904,
  2017.

\bibitem{LeT91}
M.~Ledoux and M.~Talagrand.
\newblock {\em Probability in {B}anach spaces}.
\newblock Springer-Verlag, Berlin, 1991.
\newblock Isoperimetry and processes.

\bibitem{Men15}
S.~Mendelson.
\newblock Learning without concentration.
\newblock {\em J. ACM}, 62(3):Art. 21, 25, 2015.

\bibitem{MoH15}
J.~Mo and R.~W. Heath.
\newblock Capacity analysis of one-bit quantized {MIMO} systems with
  transmitter channel state information.
\newblock {\em IEEE transactions on signal processing}, 63(20):5498--5512,
  2015.

\bibitem{OyR15}
S.~Oymak and B.~Recht.
\newblock Near-optimal bounds for binary embeddings of arbitrary sets.
\newblock {\em CoRR}, abs/1512.04433, 2015.

\bibitem{PlV13lin}
Y.~Plan and R.~Vershynin.
\newblock One-bit compressed sensing by linear programming.
\newblock {\em Comm. Pure Appl. Math.}, 66(8):1275--1297, 2013.

\bibitem{PlV13}
Y.~Plan and R.~Vershynin.
\newblock Robust 1-bit compressed sensing and sparse logistic regression: a
  convex programming approach.
\newblock {\em IEEE Trans. Inform. Theory}, 59(1):482--494, 2013.

\bibitem{PlV14}
Y.~Plan and R.~Vershynin.
\newblock Dimension reduction by random hyperplane tessellations.
\newblock {\em Discrete Comput. Geom.}, 51(2):438--461, 2014.

\bibitem{Rob62}
L.~Roberts.
\newblock Picture coding using pseudo-random noise.
\newblock {\em IRE Transactions on Information Theory}, 8(2):145--154, 1962.

\bibitem{Rom09}
J.~Romberg.
\newblock Compressive sensing by random convolution.
\newblock {\em SIAM J. Imaging Sci.}, 2(4):1098--1128, 2009.

\bibitem{Tal14}
M.~Talagrand.
\newblock {\em Upper and lower bounds for stochastic processes}.
\newblock Springer, Heidelberg, 2014.

\end{thebibliography}
\end{document}